\documentclass[12pt]{elsarticle}

\makeatletter
\def\ps@pprintTitle{%
 \let\@oddhead\@empty
 \let\@evenhead\@empty
 \def\@oddfoot{}%
 \let\@evenfoot\@oddfoot}
\makeatother

\usepackage{amsmath,amsthm}
\usepackage{url}
\usepackage{enumerate}
\usepackage{xspace}
\usepackage{graphicx}
\usepackage{graphics}
\usepackage{proof}
\usepackage{algorithm}
\usepackage[noend]{algpseudocode}

\usepackage{enumitem, color}
\setitemize{itemsep=0pt,topsep=4pt,parsep=4pt,partopsep=4pt}
\setenumerate{itemsep=0pt,topsep=4pt,parsep=4pt,partopsep=4pt}
\usepackage{tikz}
\usetikzlibrary{positioning}
\usetikzlibrary{arrows}
\usepackage{mathtools}
\usepackage{tikz}
\usetikzlibrary{arrows}
\usepackage{caption}
\usepackage{float}
\usepackage{comment}
\usepackage{enumerate}
\usepackage{booktabs}
\usepackage{thmtools,thm-restate}

\usepackage{environ}

\newcommand{\repeattheorem}[1]{%
  \begingroup
  \renewcommand{\thetheorem}{\ref{#1}}%
  \expandafter\expandafter\expandafter\theorem
  \csname reptheorem@#1\endcsname
  \endtheorem
  \endgroup
}

\NewEnviron{reptheorem}[1]{%
  \global\expandafter\xdef\csname reptheorem@#1\endcsname{%
    \unexpanded\expandafter{\BODY}%
  }%
  \expandafter\theorem\BODY\unskip\label{#1}\endtheorem
}

\usetikzlibrary{decorations.pathreplacing,calc}

\newcommand{\classFont}[1]{\protect\ensuremath{\mathsf{#1}}\xspace}

\newcommand{\PSPACE}{\classFont{PSPACE}}

\newcommand{\sub}{\subseteq}
\newcommand{\tuple}[1]{\vec{#1}}

\newcommand{\Dom}{\textrm{Dom}}

\newcommand{\C}{\mathcal{C}}

\newcommand{\N}{\mathbb{N}}

\newcommand\re[2]{{
  \left.\kern-\nulldelimiterspace 
  #1 
  \vphantom{\big|} 
  \right|_{#2} 
  }}

\newcommand{\I}{\mathcal{I}}

\newcommand{\Si}{\Sigma}
\newcommand{\Sig}[1]{\Si_{\rm #1}}
\newcommand{\si}{\sigma}

\newcommand{\tn}{t_{\rm new}}
\newcommand{\vn}{v_{\rm new}}
\newcommand{\suc}{\textrm{S}}

\def\boto{\mkern1.5mu\bot\mkern2.5mu}

\newcommand{\tax}[1]{\mathcal{I}{#1}}

\newcommand{\iax}[1]{\mathcal{U}{#1}}
\newcommand{\uax}[1]{\mathcal{UI}{#1}}
\newcommand{\cax}[1]{\mathcal{C}_{#1}}
\newcommand{\fiax}{\mathcal{FI}}
\newcommand{\fax}[1]{\mathcal{F}{#1}}
\newcommand{\fmodels}{\models_{\rm fin}}
\newcommand{\scc}{\textrm{scc}}
\newtheorem{theorem}{Theorem}
\newtheorem{lemma}[theorem]{Lemma}

\newtheorem{corollary}[theorem]{Corollary}

\newtheorem{definition}[theorem]{Definition}

\newtheorem{example}[theorem]{Example}

\usepackage{amssymb}

\begin{document}

\begin{frontmatter}



\title{On the Interaction of Functional and Inclusion Dependencies with Independence Atoms\footnote{Some of our results were presented at the 23rd International Conference on Database Systems for Advanced Applications (DASFAA 2018) and the 21st International Conference on Logic for Programming, Artificial Intelligence and Reasoning (LPAR 2017)}}


\author[Suomi]{Miika Hannula}
\ead{miika.hannula@helsinki.fi}
\author[Suomi]{Juha Kontinen}
\ead{juha.kontinen@helsinki.fi}
\author[Aotearoa]{Sebastian Link}
\ead{s.link@auckland.ac.nz}
\address[Suomi]{Department of Mathematics and Statistics, University of Helsinki, Helsinki, Finland}
\address[Aotearoa]{School of Computer Science, University of Auckland, New Zealand}

\begin{abstract}
Infamously, the finite and unrestricted implication problems for the classes of i) functional and inclusion dependencies together, and ii) embedded multivalued dependencies alone are each undecidable. Famously, the restriction of i) to functional and unary inclusion dependencies in combination with the restriction of ii) to multivalued dependencies yield implication problems that are still different in the finite and unrestricted case, but each are finitely axiomatizable and decidable in low-degree polynomial time. An important embedded tractable fragment of embedded multivalued dependencies are independence atoms that stipulate independence between two attribute sets. We establish a series of results for implication problems over subclasses of the combined class of functional and inclusion dependencies as well as independence atoms. One of our main results is that both finite and unrestricted implication problems for the combined class of independence atoms, unary functional and unary inclusion dependencies are axiomatizable and decidable in low-degree polynomial time.
\end{abstract}

\end{frontmatter}

\section{Introduction}

Databases represent information about some domain of the real world. For this purpose, data dependencies provide the main mechanism for enforcing the semantics of the given application domain within a database system. As such, data dependencies are essential for most data management tasks, including conceptual, logical and physical database design, query and update processing, transaction management, as well as data cleaning, exchange, and integration. The usability of a class $\mathcal{C}$ of data dependencies for these tasks depends critically on the computational properties of its associated implication problem. The implication problem for $\mathcal{C}$ is to decide whether for a given finite set $\Sigma\cup\{\varphi\}$ of data dependencies from $\mathcal{C}$, $\Sigma$ implies $\varphi$, that is, whether every database that satisfies all the elements of $\Sigma$ also satisfies $\varphi$. If we require databases to be finite, then we speak of the finite implication problem, and otherwise of the unrestricted implication problem. While the importance of data dependencies continues to hold for new data models, the focus of this article is on the finite and unrestricted implication problems for important classes of data dependencies in the relational model of data. In this context, data dependency theory is deep and rich, and dedicated books exist \cite{DBLP:series/synthesis/2012Greco,thalheim:1991}.

Functional and inclusion dependencies constitute the most commonly used classes of data dependencies in practice. In particular, functional dependencies (FDs) are more expressive than keys, and inclusion dependencies (INDs) are more expressive than foreign keys, thereby capturing Codd's principles of entity and referential integrity, respectively, on the logical level. An FD $R:X\rightarrow Y$ with attribute subsets $X,Y$ on relation schema $R$ expresses that the values on attributes in $Y$ are uniquely determined by the values on attributes in $X$. In particular, $R:X\rightarrow R$ expresses that $X$ is a key for $R$. For example, on schema \textsc{Patient}=\{\textit{p\_id},\textit{p\_name}\} the FD $\textsc{Patient}:\textit{p\_id}\rightarrow\textit{p\_name}$ expresses that the \textit{id} of a patient uniquely determines the name of the patient, and on schema  \textsc{Test}=\{\textit{t\_id},\textit{t\_desc}\} the FD $\textsc{Test}:\textit{t\_id}\rightarrow\textit{t\_desc}$ expresses that the \textit{id} of a medical test uniquely determines the description of a test. An inclusion dependency (IND) $R[A_1,\ldots,A_n]\subseteq R'[B_1,\ldots,B_n]$, with attribute sequences $A_1,\ldots,A_n$ on $R$ and $B_1,\ldots,B_n$ on $R'$, expresses that for each tuple $t$ over $R$ there is some tuple $t'$ over $R'$ such that for all $i=1,\ldots,n$, $t(A_i)=t'(B_i)$ holds. If $n=1$ we call the IND unary (UIND). For example, on schema \textsc{Heart}=\{\textit{p\_id},\textit{p\_name},\textit{t\_id}\} the unary IND
$\textsc{Heart}[\textit{t\_id}]\subseteq \textsc{Test}[\textit{t\_id}]$ expresses that each id of a test that is performed on patients to diagnose a heart disorder must reference the (unique) id of a medical test on schema \textsc{Test}.

A fundamental result in dependency theory is that the unrestricted and finite implication problems for the combined class of FDs and INDs differ and each is undecidable \cite{chandra85,mitchell83,mitchell83a}. Interestingly, for the expressive subclass of FDs and UINDs, the unrestricted and finite implication problems still differ but each are axiomatizable and decidable in low-degree polynomial time \cite{CosmadakisKV90}.

Another important expressive class of data dependencies are \textit{embedded multivalued dependencies} (EMVDs). An EMVD $R:X\rightarrow Y\boto Z$ with attribute subsets $X,Y,Z$ of $R$ expresses that the projection $r[XYZ]$ of a relation $r$ over $R$ on the set union $XYZ$ is the join $r[XY]\bowtie r[XZ]$ of its projections on $XY$ and $XZ$. Another fundamental result in dependency theory is that the unrestricted and finite implication problems for EMVDs differ, each is not finitely axiomatizable \cite{ParkerP80} and each is undecidable \cite{herrmann:2006,Herrmann06}. An important fragment of EMVDs are multivalued dependencies (MVDs), which are a class of full dependencies in which $XYZ$ covers the full underlying set $R$ of attributes. In fact, MVDs are the basis for Fagin's fourth normal form \cite{fagin77,DBLP:conf/caise/WeiL19}. For the combined class of FDs, MVDs, and UINDs, finite implication is axiomatizable and decidable in cubic time, while unrestricted implication is also axiomatizable and decidable in almost linear time \cite{CosmadakisKV90,Kanellakis90}.

Another expressive known fragment of EMVDs that is computationally friendly is the class of \textit{independence atoms} (IAs). IAs are EMVDs $R:X\rightarrow Y\boto Z$ where $X=\emptyset$, i.e. expressing that $r[YZ]=r[Y]\bowtie r[Z]$ holds. IAs are denoted by $Y\boto Z$. In our example, the IA $\textit{p\_id}\boto\textit{t\_id}$ on schema \textsc{Heart} expresses that all patients that are tested for a heart disorder undergo all tests for this disorder. For the class of IAs, the finite and unrestricted implication problems coincide, they are finitely axiomatizable and decidable in low-degree polynomial time \cite{KontinenLV13}. Besides their attractive computational features, IAs are interesting for a variety of other reasons: (i) Database researchers studied them as early as 1976~\cite{DBLP:conf/mfcs/Cadiou76}, with continued interest over the years \cite{DBLP:journals/tods/Delobel78,HannulaKL16,KontinenLV13,paredaens:1980}. (ii) Geiger, Paz, and Pearl studied IAs in a probabilistic setting \cite{geiger:1991} where they constitute an important fragment of conditional independencies, which form the foundation for Markov and Bayesian networks. (iii) IAs occur naturally in database practice. For example, the cross product between various tables is computed by the \texttt{FROM} clause in SQL. Naturally, a variety of IAs hold on the resulting table. The choice of the best query plan depends typically on the correct estimation of cardinalities for intermediate results. For efficiency purposes, the estimates typically assume independence between attribute columns \cite{DBLP:conf/vldb/PoosalaI97}. Knowing which independencies actually hold, could replace cardinality estimation by exact cardinalities and therefore result in better query plans. We acknowledge that the independence statements required for these types of optimizations are mostly restricted to hold for specific combinations of values on the given attributes. However, to understand such expressive independence atoms, we first need to understand the more basic ones. This is therefore an exciting area of future research that will be influenced by the results we derive in the current article. More recently, Olteanu and Zavodny \cite{OlteanuZ15} studied succinct representations of relational data by employing algebraic factorizations using distributivity of Cartesian products over unions. Not surprisingly, one of the core enabling notions of the factorizations is that of independence. (iv) In fact, the concept of independence is fundamental to areas as diverse as causality, bound variables in logic, random variables in statistics, patterns in data, the theory of social choice, Mendelian genetics, and even some quantum physics \cite{DBLP:journals/eatcs/Abramsky14,AbramskyGK13}. In a recent response, the study of logics with IAs as atoms of the language has been initiated \cite{GradelV13}.

Given the usefulness of EMVDs, FDs, and INDs for data management, given their computational barriers, and given the attractiveness of IAs as a tractable fragment of EMVDs, it is a natural question to ask how IAs, FDs, and INDs interact. We aim at helping address this current gap in the existing rich theory of relational data dependencies. Adding further to the challenge it is important to note that IAs still form an embedded fragment of EMVDs, in contrast to MVDs which are a class of full dependencies. Somewhat surprisingly, already the interaction of IAs with just keys is intricate \cite{HannulaKL14,HannulaKL16}. For example, unrestricted implication is finitely axiomatizable but finite implication is not for keys and unary IAs (those with singleton attribute sets), while the finite and unrestricted implication problems coincide and enjoy a finite axiomatization for IAs and unary keys (those with a singleton attribute set).


\begin{table*}[t]
\caption{Subclasses of FD+IND+IA. We write ``ui'' and ``fi''   for unrestrited and finite implication, respectively. }\label{newresults}
\begin{center}
\scalebox{.725}{
\begin{tabular}{llll}
\toprule
	class &  ui = fi  & complexity: ui / fi & finite axiomatization: ui / fi  \\\midrule
	FD  & yes \cite{armstrong74} & linear time \cite{BeeriB79}  & yes ($2$-ary) \cite{armstrong74}\\
	IND &  yes \cite{CasanovaFP84journal}  & $\PSPACE$-complete \cite{CasanovaFP84journal}& yes ($2$-ary) \cite{CasanovaFP84journal}\\
	IA  & yes  \cite{geiger:1991,KontinenLV13,paredaens:1980}& cubic time \cite{geiger:1991,KontinenLV13}& yes ($2$-ary) \cite{geiger:1991,KontinenLV13,paredaens:1980} \\
	 IND+IA	& yes    & $\PSPACE$-complete &   yes ($3$-ary)        \\
	 	FD+IA, FD+UIA & no \cite{HannulaKL16} & ? / ? &  ? / no \\
	FD+IND  & no \cite{chandra85,mitchell83} & undecidable / undecidable \cite{chandra85,mitchell83}& no / no \cite{chandra85,mitchell83}  \\
 FD+UIND & no \cite{CosmadakisKV90} & cubic time / cubic time \cite{CosmadakisKV90} & yes / no (infinite) \cite{CosmadakisKV90} \\
  UFD+UIND & no \cite{CosmadakisKV90} & linear time / linear time \cite{CosmadakisKV90} & yes  / no (infinite) \cite{CosmadakisKV90} \\
UFD+UIND+IA & no & cubic time / cubic time & yes  / no (infinite)    \\\bottomrule
\end{tabular}
}
\end{center}
\end{table*}

In this article we make the following contributions.
\begin{enumerate}
\item For the combined class of IAs and INDs,  finite and unrestricted implication coincide and we establish a finite axiomatization. We further show that implication for this class is PSPACE-complete and fixed-parameter tractable in the maximum arity of the given dependencies. As these results already hold for INDs \cite{CasanovaFP82,CasanovaFP84journal}, adding IAs to INDs adds significant expressivity without penalties in terms of computational properties. This is in sharp contrast to adding IAs to keys \cite{HannulaKL14,HannulaKL16}.
\item For the combined class of FDs and IAs, finite and unrestricted implication differ \cite{HannulaKL14,HannulaKL16}. We show that finite implication is not finitely axiomatizable, already for binary FDs (those with a two-element attribute set on the left-hand side) and unary IAs. For the combined class of IAs and unary FDs, we show that finite and unrestricted implication coincide and establish a finite axiomatization. Hence, the situation for the combined class of FDs and IAs is more intricate than for the combined class of FDs and MVDs, where finite and unrestricted implication coincide, which enjoy an elegant finite axiomatization \cite{BeeriFH77}, and for which implication can be decided in almost linear time~\cite{Galil82}.
\item For the combined class of IAs, unary FDs, and UINDs, we prove existence of finite Armstrong relations, establish axiomatizations for their finite and unrestricted implication problems, and show that both are decidable in low-degree polynomial time. This is analogous to the results for the combined class of FDs, MVDs, and UINDs. To the best of our knowledge, the class of IAs, unary FDs, and UINDs is only the second known class for which the finite and unrestricted implication differ but both are decidable in low-degree polynomial time. The class is practically relevant as it covers arbitrary independence atoms on top of unary keys and unary foreign keys, and already unary keys and unary foreign keys occur readily in practice \cite{CosmadakisKV90}. The significant difference to FDs, MVDs, and UINDs is the more intricate interaction between FDs and IAs in comparison to FDs and MVDs. Since unary FDs and INDs frequently occur in database practice in the form of surrogate keys and foreign keys that reference them, the ability to reason efficiently about IAs, UFDs, and UINDs is good news for data management. Finally, trading in restrictions of the arity on INDs and FDs for restrictions on the arity of IAs cannot be successful:  Finite implication for unary IAs and binary FDs is not finitely axiomatizable, see 2).
\item For the combined class of IAs and INDs, and the combined class of IAs and FDs, we establish tractable conditions sufficient for non-interaction, in both the finite and unrestricted cases. For the class of IAs and INDs, the condition ensures that we can apply known algorithms for deciding implication of the individual classes of IAs and INDs, respectively, to decide implication for an input that combines both individual classes. For the general class of IAs and FDs the decidability of finite and unrestricted implication are both still open. Instances of the finite or unrestricted implication problems that meet the non-interaction conditions can therefore be decided efficiently by using already known algorithms for the sole class of IAs and the sole class of FDs.
\end{enumerate}

\textbf{Organization.} We illustrate some use cases for our work in Section~\ref{sec:motivation}. In Section~\ref{sect:preli} we present all the necessary definitions for the article. Section~\ref{sect:IND+IA} examines the axiomatic characterization of the combined class of INDs and IAs, and Section~\ref{sect:FD+IAgen} addresses the combined class of FDs and IAs. In Section~\ref{sect:UFD+UIND+IA} we focus on the combination of UFDs, UINDs, and IAs, and establish axiomatizations for their finite and unrestricted implication problems. Section~\ref{sect:poly} identifies polynomial-time criteria for the non-interaction between INDs and IAs, and also between FDs and IAs. Finally, in Section~\ref{sect:comp} we discuss the computational complexity of the implication problems considered. We conclude in Section~\ref{sec:conclusion} where we also list some direction for future work.

\section{Motivating Showcases}\label{sec:motivation}

We use a few showcases to illustrate how knowledge about independence atoms can advance data management. FDs and INDs do not require further motivation but the more we know about the interaction of IAs with FDs and INDs, the more we can advance data management.

\subsection{Advanced integrity management}

We use a simplified example to illustrate how IAs, FDs and INDs can be used to manage data integrity in databases. For this purpose, consider the four relation schemata
\begin{itemize}
\item \textsc{Patient}=\{\textit{p\_id},\textit{p\_name}\},
\item \textsc{Test}=\{\textit{t\_id},\textit{t\_desc}\},
\item \textsc{Heart}=\{\textit{p\_id},\textit{p\_name},\textit{t\_id}\}, and
\item \textsc{Disorder}=\{\textit{p\_id},\textit{t\_id},\textit{confirmed}\},
\end{itemize}

in which basic information about patients and medical tests is stored. In particular, \textsc{Heart} stores which medical tests for a specific heart disorder were performed on which patients, and \textsc{Disorder} stores all those tests performed on patients which have been diagnosed with the disorder. In addition, the following set $\Sigma$ of FDs, INDs, and IAs has been specified:
\begin{itemize}
\item $\sigma_1=\textsc{Patient}:\;\textit{p\_id}\rightarrow\textit{p\_name}$,
\item $\sigma_2=\textsc{Test}:\;\textit{t\_id}\rightarrow\textit{t\_desc}$,
\item $\sigma_3=\textsc{Heart}[\textit{p\_id},\textit{p\_name}]\subseteq \textsc{Patient}[\textit{p\_id},\textit{p\_name}]$,
\item $\sigma_4=\textsc{Heart}[\textit{t\_id}]\subseteq \textsc{Test}[\textit{t\_id}]$,
\item $\sigma_5=\textsc{Heart}:\; \textit{p\_id}\boto\textit{t\_id}$,
\item $\sigma_6=\textsc{Disorder}[\textit{p\_id}]\subseteq \textsc{Heart}[\textit{p\_id}]$,
\item $\sigma_7=\textsc{Disorder}[\textit{t\_id}]\subseteq \textsc{Heart}[\textit{t\_id}]$,
\item $\sigma_8=\textsc{Disorder}:\; \textit{confirmed}\boto\textit{confirmed}$.
\end{itemize}

Note that not all constraints need to be enforced strictly. For example, violations of $\sigma_5$ may issue alerts about patients that still have to undergo remaining tests. The IA $\textit{confirmed}\boto\textit{confirmed}$ expresses that all tuples have the same value on attribute $\textit{confirmed}$. There are a number of interesting dependencies that are implied by $\Sigma$. Firstly, the IND $\sigma_3$ together with the FD $\sigma_1$ finitely imply the FD $\sigma_9=\textsc{Heart}:\;\textit{p\_id}\rightarrow\textit{p\_name}$. In turn, the FD $\sigma_9$ and the IA $\sigma_5$ together finitely imply the IA $\sigma_{10}=\textsc{Heart}:\;\textit{p\_id,p\_name}\boto\textit{t\_id}$ and thus also $\sigma_{11}=\textsc{Heart}:\;\textit{p\_name}\boto\textit{t\_id}$. Finally, the INDs $\sigma_6$ and $\sigma_7$ and the IA $\sigma_5$ together finitely imply the IND
\[\sigma_{12}=\textsc{Disorder}[\textit{p\_id,t\_id}]\subseteq \textsc{Heart}[\textit{p\_id,t\_id}]\;.\]
In particular, the last interaction is very relevant in practice. While the two UINDs $\sigma_6$ and $\sigma_7$ do not together imply the IND $\sigma_{12}$, knowing that the IA $\sigma_5$ holds on the referenced schema, tells us that $\sigma_{12}$ also holds on the referencing schema. It may be more natural to specify $\sigma_{12}$ in the first place, instead of specifying $\sigma_6$ and $\sigma_7$, but enforcing these two unary INDs and the IA $\sigma_5$ is more efficient than enforcing the binary IND $\sigma_{12}$ and the IA $\sigma_5$~\cite{DBLP:conf/edbt/MemariL15}.

\subsection{Query optimization}

As another example for the usefulness of understanding the interaction between IAs, FDs, and INDs, we consider query optimization. The famous division operator  $\pi_{XY}(R)\div\pi_Y(R)$ returns all those $X$-values $x$ such that for every $Y$-value $y$ there is some tuple $t$ with $t(X)=x$ and $t(Y)=y$ \cite{DBLP:persons/Codd72}. The ability of the division operator to express universal quantification makes it very powerful. The validity of independence is intrinsically linked to the optimization of the division operator, as our following result suggests.

\begin{theorem}
For all relations $r$ over $R$, $\pi_{XY}(R)(r)\div\pi_Y(R)(r)=\pi_X(R)(r)$ if and only if $r$ satisfies $X\bot Y$.
\end{theorem}

\begin{proof} The division operator is defined as follows:
\begin{align*}
&\pi_{XY}(R)(r)\div\pi_Y(R)(r)=\pi_X(R)(r)-\\
&\pi_X((\pi_X(R)(r)\times\pi_Y(R)(r))-\pi_{XY}(R)(r))\;,
\end{align*}
and $r$ satisfies $X\bot Y$ if and only if $\pi_X(R)(r)\times\pi_Y(R)(r)=\pi_{XY}(R)(r)$. The result follows directly.
\end{proof}

In particular, the validity of an IA reduces the qua\-dra\-tic complexity of the division operator to a linear complexity of a simple projection \cite{DBLP:conf/pods/LeindersB05}, as illustrated next on our running example. Suppose, we would like to return the p\_id of people that have undergone all tests listed for the specific heart disorder we consider. We can express this query by a division operator as follows: $\pi_{p\_id,t\_id}(\textsc{Heart})\div\pi_{t\_id}(\textsc{Heart})$. In SQL, the query would have to use double-negation as in:

\begin{center}
\begin{tabular}{l}
SELECT H0.p\_id FROM \textsc{Heart} H0 \\
WHERE NOT EXISTS \\
\hspace*{.5cm}SELECT $\ast$ FROM \textsc{Heart} H1 \\
\hspace*{.5cm}WHERE NOT EXISTS \\
\hspace*{1cm}SELECT $\ast$ FROM \textsc{Heart} H2 \\
\hspace*{1cm}WHERE H2.t\_id = H1.t\_id AND \\
\hspace*{2cm}H2.p\_id = H0.p\_id\;; \\
\end{tabular}
\end{center}
However, if a query optimizer can notice that the IA $\sigma_{11}$ is implied by the enforced set $\Sigma$ given above, then the query can be rewritten into
\begin{center}
\begin{tabular}{l}
SELECT p\_id \\
FROM \textsc{Heart}\,;
\end{tabular}
\end{center}
While the set of our constraints is weakly acyclic~\cite{DBLP:journals/tcs/FaginKMP05}, our query is not ``path-conjunctive" and the chase \& backchase algorithm from~\cite{DBLP:conf/vldb/DeutschPT99} cannot be applied.

\subsection{Inference control}

Our second example is database security. More specifically, the aim of inference control is to protect private data under inferences that clever attacks may use to circumvent access limitations \cite{DBLP:journals/ijisec/BiskupB04}. For example, the combination of a particular patient name (say Jack) together with a particular medical examination (say angiogram) may be considered a secret, while access to the patient name and access to the medical examination in isolation may not be a secret. However, in some given context such as a procedure to diagnose some condition, all patients may need to undergo all examinations. That is, the information about the patient is independent of the information about the examination. Now, if the secret (Jack, angiogram) must not be revealed to an unauthorized user that can query the data source, then this user must not learn both: that Jack is a patient undergoing the diagnosis of the condition, and that angiogram is a medical examination that is part of the process for diagnosing the condition. Being able to understand the interaction of independence atoms with other database constraints can therefore help us to protect secrets under clever inference attacks.

\subsection{Data profiling}

Our final example is data profiling. Here we would like to demonstrate that independence atoms do occur in real-world data sets. For that purpose, we have mined some well-known publicly available data sets that have been used for the mining of other classes of data dependencies before \cite{DBLP:journals/pvldb/PapenbrockEMNRZ15}. We report the basic characteristics of these data sets in the form of their numbers of rows and columns, and list the number of maximal IAs and the maximum arity of those found. Here, an IA $X\boto Y$ is maximal in a given set of IAs if there is no other IA $V\boto W$ in the set such that $V\subseteq X$ and $W\subseteq Y$ holds. The arity of an IA is defined as the total number of attribute occurrences.

\begin{center}
\scalebox{.85}[.85]{
\begin{tabular}{c@{\hspace*{.25cm}}c@{\hspace*{.25cm}}c@{\hspace*{.25cm}}c@{\hspace*{.25cm}}c}\hline
\textit{Data set} & \textit{Number of columns} & \textit{Number of rows} & \textit{Number of IAs} & \textit{Maximum arity}\\ \hline
bridges        & 13 & 108 & 4 & 3 \\
echocardiogram & 13 & 132 & 5 & 4 \\
adult          & 14 & 48,842 & 9 & 3 \\
hepatitis      & 20 & 155 & 855 & 6 \\
horse          & 27 & 368 & 112 & 3 \\ \hline
\end{tabular}
}
\end{center}

It should be stressed that the usefulness of these IAs is not restricted to those that are semantically meaningful. For example, the optimizations for the division operator also apply to IAs that ``accidentally'' hold on a given data set. For the future, we envision that data profiling tools can also keep profiles of sophisticated notions of independence atoms. For example, in the data set \emph{hepatitis} one desires non-bias and therefore an independence atom $\emph{age}\bot\emph{sex}$ to hold. Indeed, if this independence held we would know that the number of distinct tuples in the projection of \emph{hepatitis} onto age and sex would be the product of the distinct tuples in the projection onto age and in the projections onto sex, however we have \[|\emph{hepatitis}[\emph{age},\emph{sex}]|=0.612\times|\emph{hepatitis}[\emph{age}]|\times|\emph{hepatitis}[\emph{sex}]|\;.\] This is valuable information that could be profiled.

\section{Preliminaries}\label{sect:preli}
We usually write $A,B,C,... $ for attributes, $X,Y,Z, ... $ for either sets or sequences of attributes, depending on the context. For two sets (sequences) $X$ and $Y$, we write $XY$ for their union (concatenation). Similarly, we may write $A$ instead of the single element set or sequence that consists of $A$. 
The size of a set (or length of a sequence) $X$ is denoted by $|X|$.

A \emph{relation schema} is a set of attributes $A$, each with a \emph{domain} $\Dom(A)$, and by a \emph{database schema} we denote a pairwise disjoint sequence of relation schemata. A \emph{tuple} over a relation schema $R$ is a function that maps each $A\in R$ to $\Dom(A)$. A \emph{relation} $r$ over $R$ is a non-empty set of tuples over $R$. To emphasize that  $r$ is a relation over $R$, we sometimes write  $r[R]$. A \emph{database} over a database schema $R_1, \ldots ,R_n$ is a sequence of relations $(r_1[R_1], \ldots ,r_n[R_n])$. A \emph{finite} relation over $R$ is a non-empty, finite set of tuples over $R$, and a finite database is a sequence of finite relations.
For a tuple $t$ and a relation $r$ over $R$ and $X\sub R$, $t(X)$ is the \emph{restriction} of $t$ to  $X$, and $r(X)$ is the set of all restrictions $t(X)$ where $t\in r$. If 
$X=(A_1,\ldots ,A_n)$ is a sequence of attributes, then we write $t(X)$ for  $(t(A_1), \ldots ,t(A_n))$.

We exclude empty relations from our definition. This is a practical assumption with no effect when single relation schemata are considered only. However, on multiple relations it has an effect, e.g., the rule $\uax{3}$ in Table \ref{tab-rules3} becomes unsound.

Syntax and semantics of FDs, INDs, and IAs are as follows. Let $d=(r_1[R_1], \ldots ,r_n[R_n])$ be a database. For two sequences of distinct attributes $A_1, \ldots ,A_n\in R_i$ and $B_1, \ldots ,B_n\in R_j$, $R_i[A_1\ldots A_n ]\sub R_j[B_1\ldots B_n]$ is an \emph{inclusion dependency} with semantics defined by $d\models  R_i[A_1\ldots A_n] \sub R_j[ B_1\ldots B_n]$ if for all $t\in  r_i$ there is some $t'\in r_j$ such that $t(A_1)=t'(B_1), \ldots ,t(A_n)=t'(B_n)$. For two (not necessarily disjoint) sets  of attributes $X,Y\sub R_i$,  $R_i:X\boto Y$ is an \emph{independence atom} with semantics: $d\models R_i:X\boto Y$ if for all $ t,t'\in r_i$ there exists $ t''\in r_i$ such that $t''(X)=t(X)$ and $ t''(Y)=t'(Y)$. For two sets of attributes  $X,Y\sub R_i$,  $R_i:X\to Y$ is a \emph{functional dependency} with semantics: $d\models R_i:X\to Y$ if for all $ t,t'\in r_i$, $t(X)=t'(X)$ implies $t(Y)=t'(Y)$. We may exclude  relation schemata from the notation if they are clear from the context (e.g. write $X\boto Y$ instead of $R_i: X\boto Y$). A \emph{disjoint} independence atom (DIA) is an IA $X\boto Y$ where  $X\cap Y$ is empty. 
We say that an IND is \emph{$k$-ary} if it is of the form $A_1\ldots A_k\sub B_1\ldots B_k$. An IA  $X \boto Y$ and an FD $X\to Y$ are called $k$-ary if $\max\{|X|,|Y|\}= k$. A class of dependencies is called $k$-ary if it contains at most $k$-ary dependencies.  We add ``U'' to a class name to denote its unary subclass, e.g., UIND denotes the class of all unary INDs. Similarly, for $k\geq 2$ we add ``$k$'' to a class name to denote its $k$-ary subclass. We use ``$+$'' to denote unions of classes, e.g., IND+IA denotes the class of all inclusion dependencies and independence atoms. Note that the semantics of IAs implies:
\begin{itemize}
\item[*]$d\models R_i:X\boto X$, if \\ for all $ s,s'\in r_i$ it holds that $s(X)=s'(X)$.
\end{itemize}
Hence, unary FDs of the form $\emptyset \to A$ and unary IAs of the form $A\boto A$ are also called \emph{constancy atoms} (CAs). 

The \emph{restriction} of a dependency $\sigma$ to a set of attributes $R$, written $\sigma\upharpoonright R$, is $X\cap R \to Y\cap R$ for an FD $\si$ of the form $X\to Y$, and $X\cap R\boto Y\cap R$ for an IA $\si$ of the form $X\boto Y$. If $\sigma$ is an IND of the form $A_1\ldots A_n\sub B_1\ldots B_n$ and $i_1, \ldots ,i_k$ lists $\{i=1, \ldots ,n:A_i\in R\textrm{ and }B_i\in R\}$, then $\sigma\upharpoonright R=A_{i_1}\ldots A_{i_k}\sub B_{i_1}\ldots B_{i_k}$. For a set of dependencies $\Sigma$, the restriction of $\Sigma$ to $R$, written $\Sigma\upharpoonright R$, is the set of all  $\sigma\upharpoonright R$ where $ \sigma \in \Sigma$. Let $A$ and $B$ be attributes from $R$. By $\si(R:A\mapsto B)$ we denote dependencies obtained from $\si$ by replacing any number of occurrences of $A$ with $B$.

A set $\mathfrak{R}$ of rules of the form $\si_1, \ldots ,\si_n \Rightarrow \si$ is called an \emph{axiomatization}. A rule of the previous form is called $n$-ary, and an axiomatization consisting of at most $n$-ary rules is called $n$-ary. A \emph{deduction} from a set of dependencies $\Si$ by an axiomatization $\mathfrak{R}$ is a sequence of dependencies $(\si_1, \ldots ,\si_n)$ where each $\si_i$ is either an element of $\Si$ or follows from $\si_1, \ldots ,\si_{i-1}$ by an application of a rule in $\mathfrak{R}$. In such an occasion we write $\Si\vdash_{\mathfrak{R}}\si$, or simply $\Si\vdash \si$ if $\mathfrak{R}$ is known.

Given a finite set of database dependencies $\Si\cup\{\si\}$, the (finite) unrestricted implication problem is to decide whether all (finite) databases that satisfy $\Si$ also satisfy $\si$, written $\Si\models \si$ ($\Si\fmodels \si$). An axiomatization $\mathfrak{R}$ is \emph{sound} for the unrestricted implication problem of a class of dependencies $\C$ if for all finite sets $\Si\cup\{\si\}$ of dependencies from $\C$, $\Si\vdash_{\mathfrak{R}} \si \Rightarrow \Si\models \si$; it is \emph{complete} if $ \Si\models \si\Rightarrow  \Si\vdash_{\mathfrak{R}} \si $. Soundness and completeness for finite implication are defined analogously.

We assume that all our axiomatizations are attribute-bounded. A sound and complete axiomatization is said to be \emph{attribute-bounded} if it does not introduce new attributes, i.e., any implication of $\si$ by $\Si$ can be verified by a deduction in which only attributes from $\Si$ or $\si$ appear \cite{chandra85}. It is easy to see that a 
finite (attribute-bounded) axiomatization gives rise to a decision procedure for the associated implication problem. The converse is not necessarily true; join dependencies consitute a class that is associated with a decidable implication problem, yet they lack finite axiomatization \cite{Petrov:1989}.  Consider then the class FD+IND+IA. Clearly, both sets $\{(\Si,\si)\mid \Si \models \si\}$ and $\{(\Si,\si)\mid \Si \not\fmodels \si\}$ are recursively enumerable; the first via reduction to the validity problem of first-order logic, and the second by checking through whether some finite relation satisfies $\Si\cup\{\neg \si\}$. Consequently, given a subclass $\mathcal{C}$ of FD+IND+IA, the unrestricted and finite implication problems for $\mathcal{C}$ are decidable whenever these two problems coincide.

Many of our completeness proofs utilize the chase technique (see, e.g., \cite{AbiteboulHV95}). The chase provides a general tool for reasoning about various dependencies as well as for optimizing conjunctive queries.  Given an implication problem for $\si$ by $\Si$, the starting point of the chase is a simple database falsifying $\si$. For instance, the chase for independence atoms starts with a unirelational database consisting of two rows that disagree on all attributes. Using some dedicated set of chase rules, this initial database is then completed to another database satisfying $\Si$. If the new database satisfies also $\si$, then one concludes that the implication holds.  For some classes, such as embedded multivalued dependencies, the  chase does not necessarily terminate. In those cases only a semi-decision procedure is obtained.


\textbf{Axiomatizations.} Tables \ref{tab-rules}, \ref{tab-rules2}, and \ref{tab-rules3}  present  the axiomatizations considered in this article. In Table \ref{tab-rules}, the axiomatization $\mathfrak{I}:=\{\tax{1}, \ldots ,\tax{5}\}$ is sound and complete for independence atoms alone  \cite{HannulaKL16,KontinenLV13}. The rules $\fax{1}, \fax{2} ,\fax{3}$ form the Armstrong axiomatization for functional dependencies \cite{armstrong74}, and the rules $\fiax{1}$ and $\fiax{2}$ describe simple interaction between independence atoms and functional dependencies. Table \ref{tab-rules2} depicts the sound and complete axiomatization of inclusion dependencies introduced in \cite{CasanovaFP82,CasanovaFP84journal}. Table \ref{tab-rules3} presents rules  describing interaction between inclusion dependencies and independence atoms.

We conclude this section by stating the soundness of the axioms in Tables \ref{tab-rules}, \ref{tab-rules2}, and \ref{tab-rules3}. The proof is a straightforward exercise and left to the reader. Note that soundness of $\uax{3}$ follows only if databases are not allowed to contain empty relations.

\begin{theorem}\label{soundness}
The axiomatization $\mathfrak{A}\cup\mathfrak{B}\cup\mathfrak{C}$ is sound for the unrestricted and finite implication problems of FD+IND+IA.
\end{theorem}

\begin{table}[h]
\[\fbox{$\begin{array}{c@{\hspace*{.25cm}}c}
\cfrac{}{\emptyset\boto X} & \cfrac{X\boto Y}{Y\boto X} \\
\text{(trivial independence, $\tax{1}$)} & \text{(symmetry, $\tax{2}$)}\\

\\ \\

\cfrac{X\boto YZ}{X\boto Y}&\cfrac{X\boto Y\quad XY\boto Z}{X\boto YZ}   \\
 \text{(decomposition, $\tax{3}$)}& \text{(exchange, $\tax{4}$)}\\

\\ \\

\cfrac{X\boto Y \quad Z\boto Z}{X\boto YZ} &	 \cfrac{}{XY\to Y} \\
\text{(weak composition, $\tax{5}$)} & \text{(reflexivity, $\fax{1}$)}

\\ \\

\cfrac{X \to Y \quad Y\to Z}{X\to Z}    &	  \cfrac{X\to Y}{XZ\to YZ}  \\
 \text{(transitivity, $\fax{2}$)}& \text{(augmentation, $\fax{3}$)} \\

\\ \\

 \cfrac{X\boto Y \quad X\to Y}{\emptyset\to Y}    &   \cfrac{X\boto YZ \quad Z\to V}{X\boto YZV}\\
   \text{(constancy, $\fiax{1}$)} &   \text{(composition, $\fiax{2}$)}

\end{array}$}\]
\caption{Axiomatization $\mathfrak{A}$ for FDs and IAs \label{tab-rules}. We define $\mathfrak{I}:=\{\tax{1},\ldots ,\tax{5}\}$ and $\mathfrak{A}^*:=\mathfrak{A}\setminus\{\tax{5},\fax{3}\}$.}
\vspace{-3mm}
\end{table}


\section{Independence Atoms and Inclusion Dependencies}\label{sect:IND+IA}
In this section we establish a set of inference rules that is proven sound and complete for the unrestricted and finite implication problems of independence atoms and inclusion dependencies. This axiomatization consists of rules $\mathfrak{C}$ describing interaction between the two classes (see Table \ref{tab-rules3}) and two sets $\mathfrak{I}:=\{\tax{1},\ldots ,\tax{5}\}$ and $\mathfrak{B}$ of complete rules for both classes in isolation (see Tables \ref{tab-rules} and \ref{tab-rules2}, resp.).
Furthermore, as a consequence of the completeness proof we obtain that the finite and unrestricted implication problems coincide for IND+IA. In addition, our completeness proof enables us to construct Armstrong databases for this class of constraints, and to simplify the implication problem for a subclass of IND+IA.

\subsection{Axiomatization}

\begin{table}[h]
\[\fbox{$\begin{array}{c}

\cfrac{}{R[X] \sub R[X] }
\\
\text{(reflexivity, $\mathcal{U}1$)}
 \\ \\

 \cfrac{R[X]\sub R'[Y]  \quad R'[Y]  \sub R''[Z] }{R[X]\sub  R''[Z]}
 \\
 \text{(transitivity, $\mathcal{U}2$)}

 \\ \\

\cfrac{R[A_1 \ldots A_n] \sub R'[B_1\ldots B_n] }{R[A_{i_1}\ldots A_{i_m}] \sub R'[B_{i_1}\ldots B_{i_m}]}{(^*)}
\\
\text{(projection and permutation, $\mathcal{U}3$)}\\
\text{($^*$) $i_j$ are pairwise distinct and from $\{1, \ldots ,n\}$}  

\end{array}$}\]
\caption{Axiomatization  $\mathfrak{B}$ for INDs \label{tab-rules2}}
\vspace{-3mm}
\end{table}
\begin{table}[h]
\[\fbox{$\begin{array}{c}

\cfrac{R[X] \sub R'[Z ] \quad R[Y] \sub R'[W] \quad R'[Z \boto W]}{R[X Y]\sub R'[Z W]} \\
\text{(concatenation, $\mathcal{UI}1$)}

\\ \\

\cfrac{R[X Y] \sub R'[ZW] \quad R'[ZW]\sub R[XY] \quad R'[Z\boto W]}{R[X\boto Y]}
\\
\text{(transfer, $\mathcal{UI}2$)}

 \\ \\

 \cfrac{R[X]\sub R'[Y] \quad R': Y \boto Y}{R'[Y] \sub R[X]}

\\
\text{(symmetry, $\mathcal{UI}3$)}

\\ \\
\cfrac{R[X]\sub R'[Y] \quad R':Y\boto Y}{R:X \boto X}

\\
\text{(constancy, $\mathcal{UI}4$)}

\\ \\
\cfrac{R[A]\sub R'[C]\quad R[B]\sub R'[C]\quad R': C \boto C\quad \sigma}{\sigma(R:A\mapsto B)}

\\
\text{(equality, $\mathcal{UI}5$)}

\end{array}$}\]
\caption{Axiomatization $\mathfrak{C}$ for IAs and INDs \label{tab-rules3}}
\vspace{-3mm}
\end{table}
We start with the following simplifying lemma which  reduces one finite IND+IA-implication problem  to another that is not associated with any constancy atoms, i.e., IAs of the form $X\boto X$. Note that we write $\Sigma \dashv\vdash \Sigma'$ if $\Sigma \vdash\Sigma'$ and $\Sigma' \vdash\Sigma$.

\begin{lemma}\label{help}
Let $\Sigma$ be a set of IAs and INDs over schema $R_1, \ldots ,R_n$, and let $\C:=\bigcup_{i=1}^n\{A\in R_i\mid \Sigma \vdash R_i:A \boto A\}$. Let $\Sigma_0$ and $\sigma_0$ be the restrictions of $\Sigma$ and $\sigma$ to the attributes not in $\C$, and let $\Sigma_1$ be obtained from $\Si$ by
\begin{enumerate}[label=\emph{(\arabic*)}]
\item replacing  $R_i:X\boto Y\in \Sigma$ with $R_i:X\setminus \C\boto Y\setminus \C$,
 \item adding $R_i:A_1\ldots A_{j}\boto A_{j+1}\ldots A_m(R_i\setminus \C) $, where  $A_1, \ldots ,A_m$ is some   enumeration of $R_i\cap \C$ and $j=1, \ldots ,m$,
 \item adding $R_j[B]\sub R_i[A]$ if $\Sigma\vdash R_i[A]\sub R_j[B]\wedge R_j:B \boto B$.
  \end{enumerate}
  Then $\Sigma \dashv\vdash\Sigma_0\cup\Sigma_1 \cup \{R_i:  R_i\cap \C\boto R_i\cap \C\mid  i=1, \ldots ,n\}$ and
\begin{enumerate}[label=\emph{(\roman*)}]
\item $\sigma$ is an IA: $\Sigma \fmodels \sigma \Rightarrow \Sigma_0 \fmodels \sigma_0$,
\item $\sigma$ is an IND: $\Sigma \fmodels \sigma \Rightarrow \Sigma_1\fmodels \sigma$.
\end{enumerate}
\end{lemma}

\begin{proof}

Clearly we have that $\Sigma \dashv\vdash\Sigma_0\cup\Sigma_1 \cup\{R_i:  R_i\cap \C\boto R_i\cap \C\mid  i=1, \ldots ,n\}$. For claim (i) note that any finite  database $d=(r_1, \ldots ,r_n)$ satisfying $\Sigma_0 \cup\{\neg \sigma_0\}$ can be extended to a one satisfying $\Sigma \cup\{\neg \sigma\}$ by replacing in each $r_i\in d$ each tuple $t$ with all  tuples $t'$ such that $t'(A)=0$ for $A\in R_i\cap \C$ and $t'(A)\in \{0,t(A)\}$ for  $A\in R_i\setminus \C$, where $0$ is a value not appearing $d$. 

Next we show claim (ii). Assuming a finite database \[d'=(r'_1[R_1], \ldots ,r'_n[R_n])\] satisfying $\Sigma_1\cup\{\neg \sigma\}$ for $\si$ of the form $R_l[X] \sub R_{l'}[Y]$, we construct a finite database $d=(r_1[R_1], \ldots ,r_n[R_n])$ satisfying $\Sigma\cup\{\neg \sigma\}$. Let $t\in r'_l$ be such that  $t(X)\neq t'(Y)$ for all $t'\in r'_{l'}$. Let $t_0$ be an extension of $t(R_l\cap \C)$ to $\C$ such that, for $A\in R_i\cap \C$ and  $B\in R_l\cap  \C$, $t_0(A)=t(B)$ if   $\Sigma \vdash R_i[A]\sub R_l[B]$, and otherwise $t_0(A)$ is any value from $r'_i(A)$.  Note that we may assume without losing generality that $t_0$ is well-defined, i.e., for no distinct $B,B'\in R_l\cap \C$, $\Sigma \vdash R_l[B]\sub R_l[B']$.
For this, define an equivalence class $\sim$ on $R_l\cap\C$ such that $B\sim B'$ if $\Sigma \vdash R_l[B]\sub R_l[B']$. Using $\uax{5}$ it is then straightforward to show that $\Sigma \fmodels \sigma\Rightarrow\Sigma^* \fmodels \sigma^*$ and   $\Sigma^*\vdash \sigma^*\Rightarrow \Sigma\vdash \sigma$ where  $\Sigma^* \cup\{\sigma^*\}$ is the set of constraints obtained from $\Sigma\cup \{\sigma\}$ by replacing attributes in $R_l\cap\C$ with their equivalence classes.

Now, define $r_i:=r'_i(R_i\setminus \C)\times \{t_0(R_i\cap \C)\}$, for $i=1, \ldots ,n$. Since $t\in r_l$ and   $r_{l'}\sub r'_{l'}$ by items (2,3) and  the construction, we obtain that $d'\not\models R_l[X] \sub R_{l'}[Y]$.  It also easy to see by the construction that all IAs in $\Si$ remain true in $d$. Assume then that $R_i[X_1\ldots X_m]\sub R_j[Y_1\ldots Y_m]\in \Sigma$, and let $t\in r_i$. 
Since $Y_i\in \C$ implies $X_i\in \C$ and $t_0(X_i)=t_0(Y_i)$, we can assume that $Y_1, \ldots ,Y_m\not\in \C$. Hence, $r_{j}(Y_1 \ldots Y_m) =r'_{j}(Y_1 \ldots Y_m) $. Again, $r_i\sub r'_i$ by (2,3) and the construction, and $d'\models R_i[X_1\ldots X_m]\sub R_j[Y_1\ldots Y_m]$; hence we obtain that $d\models R_i[X_1\ldots X_m]\sub R_j[Y_1\ldots Y_m]$. This concludes case (ii) and the proof.
\end{proof}

The following lemma will be also helpful in the sequel.
\begin{lemma}\label{lemma:help3}
$XY\boto UV$ can be deduced from $XU\boto YV$, $X\boto U$, and  $Y\boto V$ by rules  $\tax{2},\tax{3},\tax{4}$.
\end{lemma}
\begin{proof}
The following deduction shows the claim:\\

\infer[^{\tax{2}}]
	{XY \boto UV}
	{\infer[^{\tax{4}}]
     		{UV\boto XY}
     		{\infer[^{\tax{2}}]{UV\boto Y}{
     			\infer[^{\tax{3}}]{Y\boto UV}{
     				\infer[^{\tax{4}}]{Y\boto XUV}{Y\boto V &
     					\infer[^{\tax{2}}]{YV\boto XU}{XU\boto YV}}}}&
     		\infer[^{\tax{2}}]{YUV\boto X}{
     			\infer[^{\tax{4}}]{X\boto YUV}{X\boto U & XU\boto YV }}}}
\vspace{-2mm}
\end{proof}

Using the previous lemmata we can now state the completeness result.  The proof is divided into three subcases in which either CA, IND, or IA consequences are considered. By Lemma \ref{help} we may consider only IND+DIA-implication in the latter two cases. These cases are proved by a chase argument that  generalizes the completeness proof of IND-axioms presented in \cite{CasanovaFP84journal}.

\begin{theorem}\label{theorem:ind+ia}
The axiomatization $\mathfrak{I}\cup\mathfrak{B}\cup\mathfrak{C}$ is sound and complete for the unrestricted and finite implication problems of IA+IND.
\end{theorem}
\begin{proof}
 By Theorem \ref{soundness} the axiomatization is sound.
For completeness with respect to both implication problems, it suffices to show that finite implication entails derivability. For this, notice that unrestricted implication entails finite implication. Hence, assume that $\Sigma \models_{\rm FIN} \sigma$ for a finite set  $\Sigma \cup\{\sigma\}$ of IAs and INDs over database schema $R_1, \ldots ,R_n$. Let $\mathcal{C}:=\bigcup_{i=1}^n\{A\in R_i\mid \Sigma \vdash R_i:A \boto A\}$.  By $\tax{1}-\tax{3}$ we may assume without losing generality that $\sigma$ is either a CA, a DIA, or an IND. Next we show that $\Si\vdash \si$ in these three cases.\\
\textbf{1) $\sigma$ is a constancy atom.}  Assume that $\sigma$ is of the form $R_l: A \boto  A$, and assume to the contrary that $\Sigma \not\vdash \sigma$.  First let $\mathcal{I}$ be the set of attributes $B$ for which there is $i=1, \ldots ,n$  such that $\Sigma \vdash R_l[A]\sub R_i[B]$. Then let  $d=(r_1, \ldots ,r_n)$ be the database where $r_i:= {}^{R_i\cap \I}\{0,1\}\times {}^{R_i\setminus \I}\{0\}$. We show  that $d\models \Sigma\cup\{\neg \sigma\}$ which contradicts the assumption that $\Sigma \models_{\rm FIN} \sigma$. It is easy to see that $ d$ satisfies 
$\neg \sigma$. Furthermore, $d$ satisfies any $R_i[A_1\ldots A_m] \sub R_j[B_1\ldots B_m]$ from $ \Sigma$ 
because $A_i\in\I \Rightarrow B_i\in \I$ by  $\iax{2}$. 

Assume then that   $R_i:XZ\boto YZ\in \Sigma$ where $X$ and $Y$ are disjoint. By the construction, $d\models R_i:X\boto Y$, so it suffices to show that $d\models R_i:B\boto B$ for $B\in Z$. If $d\not\models R_i:B \boto B$, then by the construction $\Sigma \vdash R_l[A]\sub R_i[B]$. Moreover by $\uax{3}$, $\Sigma \vdash R_i[B]\sub R_l[A]$, and by $\uax{4}$, $\Sigma \vdash R_l:A\boto A$, contrary to the assumption. Hence $d\models R_i:B\boto B$ which concludes the proof of $d\models \Sigma \cup\{\neg\sigma\}$ and the case of $\sigma$ being a constancy atom.\\
\textbf{2) $\sigma$ is a disjoint independence atom.}
 Assume that $\sigma$ is a DIA of the form $R_l:A_1\ldots A_h \boto A_{h+1}, \ldots ,A_{h+k}$. By Lemma \ref{help} we may assume that $\Sigma$ is a set of  DIAs and INDs.
 Define first a database
 $d_0=(r_1[R_1], \ldots ,r_n[R_n])$
 such that
 \begin{itemize}
 \item $r_l=\{s,s'\}$ where $s$ and $s'$ map all attributes in $R_l$ to $0$ except that $s(A_i)=i$ for $i=1, \ldots ,h$ and  $s'(A_i) =i$ for  $i=h+1, \ldots ,h+k$;
 \item $r_i=\{u\}$ where $u$  maps all attributes in $R_i$ to $0$, for  $i\neq l$.
 \end{itemize}
The idea is to  extend $d_0$ to a database $d=(r_1,\ldots ,r_n)$ such that $d\models \Sigma$ and
\begin{itemize}
\item[*]  if $t\in r_i$ is such that $t(B_1)=i_1  , \ldots ,t(B_m)=i_m  $ and $0<  i_1< \ldots < i_m $, then $\Sigma \vdash R_l[A_{i_1}\ldots A_{i_m}]\sub R_i[B_1\ldots B_m]$.
\end{itemize}
We let $d$ be the result of chasing $d_0$ by $\Si$ over the following two chase rules, i.e., $d$ is obtained by applying rules (i-ii) to $d_0$ repeatedly until this is no more possible.
\begin{enumerate}[label=\emph{(\roman*)}]
\item Assume that $R[X]\sub R'[Y]\in \Sigma$ and $t\in r[R]$ is such that for no $t'\in r'[R']$, $t(X)=t'(Y)$. Then extend $r'$ with $t_{\rm new}$ that maps $Y$ pointwise to $t(X)$ and otherwise maps attributes in $R'$ to $0$.
\item Assume that $R:X \boto Y\in \Sigma$ and  $t,t'\in r[R]$ are such that for  no $t''\in r[R]$,  $t''(X) =t(X)$ and $t''(Y)=t'(Y)$. Then extend $r$ with $t_{\rm new}$ that agrees with $t$ on $X$, with $t'$ on $Y$, and maps every other attribute in $R$ to $0$.
\end{enumerate}
Note that since the range of the assigned values is finite, the process terminates. Hence, $d$ is a finite model of $\Si$, and therefore by the assumption  it satisfies $ \si$. It is  also straightforward to verify, using $\iax{1},\iax{3}$ at the initial stage, $\iax{2},\iax{3}$ in  items (i,ii), and $\tax{2},\tax{3},\uax{1}$ in item (ii), that $*$ is satisfied.

Since $d$ satisfies $\si$, we find a tuple from $r_l$  mapping $A_i$ to $i$ for $i=1, \ldots ,n$. 
It thus suffices to show that, given any sequence $\tuple d=(d_0, \ldots ,d_m)$, where $d_{i+1}$ is obtained from $d_i$ by applying (i) or (ii), and any $S\sub \{1, \ldots ,h+k\}$, if there is $t\in r_l$ for $d_m=(r_1, \ldots ,r_n)$ such that $t(A_i)=i$ for $i\in S$, then $\Si\vdash R_l: S\cap A_1\ldots A_h\boto S\cap A_{h+1}\ldots A_{h+k}$. We show this claim by induction on the number of applications of (ii) in $\tuple d$.

Assume first that no application of (ii) occurs. Then $t$ cannot  combine any values from both $s$ and $s'$ and hence $S$ cannot intersect both $\{1, \ldots ,h\}$ and $\{h+1, \ldots ,h+k\}$. Consequently, $R_l:S\cap A_1\ldots A_h\boto S\cap A_{h+1}\ldots A_{h+k}$ is derivable by  $\tax{1}$ and $\tax{2}$.

Let us then show the induction step. Assume that the claim holds for $j$ applications of (ii). We prove the claim for $S=\{1, \ldots ,h+k\}$ (the general case $S\subseteq \{1, \ldots ,h+k\}$ is analogous) and $\tuple d=(d_0, \ldots ,d_m)$ in which the number of applications of (ii) is $j+1$. Let $d_k$ be the database obtained by applying (ii) for the last time, say with regards to some $R: X\boto  Y$ and tuples $t,t',t_{\rm new}$. Without loss of generality we may assume that the sequence $(d_{k+1},\ldots ,d_m)$ is obtained by a chain of applications of (i) copying $1, \ldots ,h+k$ from $t_{\rm new}(B_1), \ldots ,t_{\rm new}(B_{h+k})$ to $t_0(A_1), \ldots ,t_0(A_{h+k})$, for some $B_1, \ldots  ,B_{h+k}\in  X  Y$ and $t_0$ from database $d_m$. Otherwise, step $k$ can be omitted and the claim follows by induction assumption. Now, using repeatedly $\iax{2},\iax{3}$ we obtain that
\begin{equation}\label{eq1}
\Sigma \vdash R_i[B_1\ldots B_{h+k}]\sub R_l[ A_1\ldots A_{h+k}].
\end{equation}
Moreover, since $*$ is satisfied with regards to $d_m$ we have
\begin{equation}\label{eq2}
\Sigma \vdash R_l[A_1\ldots A_{h+k}]\sub R_i[ B_1\ldots B_{h+k}].
\end{equation}
Let us then define another sequence of databases \[\tuple d'=(d_0, \ldots ,d_{k-1},d'_{k+1},\ldots ,d'_m)\] in which $t_{\rm new}$ does not appear and $(d'_{k+1},\ldots ,d'_m)$ is obtained by a chain of applications of (i) copying  $t(B_1), \ldots ,t(B_{h+k})$ to $t_1(A_1), \ldots ,t_1(A_{h+k})$ for some $B_1, \ldots  ,B_{h+k}\in  X  Y$ and $t_1$ from database $d'_m$. Let \[B_{i_1} \ldots  B_{i_a} B_{i_{a+1}} \ldots B_{i_b} B_{i_{b+1}}\ldots B_{i_{c}} B_{i_{c+1}} \ldots B_{i_{d}}\] relist  $B_1\ldots B_{h+k}$ so that
\begin{itemize}
 \item 
 $B_{i_1}, \ldots , ,B_{i_b}\in  X$ and $B_{i_{b+1}}, \ldots ,B_{i_d}\in  Y$
  \item
     $\{i_1, \ldots ,i_a, i_{b+1}, \ldots ,i_{c}\} = \{1, \ldots ,h\}$ and\\
      $\{i_{a+1}, \ldots ,i_b, i_{c+1}, \ldots ,i_{d}\} =\{ h+1, \ldots ,h+k\}$.
  \end{itemize}
Since $\Si \vdash R:B_{i_1} \ldots B_{i_b}\boto B_{i_{b+1}}\ldots B_{i_d}$ by $\tax{2}$ and $\tax{3}$, we obtain using $\eqref{eq1}$, $\eqref{eq2}$,  $\iax{3}$, and $\uax{2}$ that
\begin{equation}\label{eq3}
\Si \vdash R_l:A_{i_1} \ldots A_{i_b}\boto A_{i_{b+1}}\ldots A_{i_d}.
\end{equation}
 Furthermore, we observe that
 \[t_1(A_{i_1}, ,\ldots ,A_{i_b})=t(B_{i_1}, \ldots ,B_{i_b})=t_{\rm new}(B_{i_1}, \ldots ,B_{i_b})=(i_1, \ldots ,i_b).\]
 Since $\tuple d'$ contains only $j$ applications of (ii), it follows by induction assumption that $\Si\vdash R_l:A_{i_1}\ldots A_{i_{a}}\boto A_{i_{a+1}}\ldots A_{i_{b}}$. It can be shown by an analogous argument that $\Si\vdash R_l:A_{i_{b+1}}\ldots A_{i_{c}}\boto A_{i_{c+1}}\ldots A_{i_{d}}$. By Lemma \ref{lemma:help3} these two and \eqref{eq3} imply that $\Sigma \vdash R:A_1\ldots A_h\boto A_{h+1}\ldots A_{h+k}$. This concludes the induction proof and the case of $\sigma$ being a disjoint IA.\\

\textbf{3) $\sigma$ is an inclusion dependency.}
Assume that $\sigma$ is an IND of the form $R_l[X] \sub R_{l'}[ Y]$. By Lemma \ref{help} we may assume that $\Sigma$ is a set of  IAs and disjoint INDs. We let $d=(r_{1}, \ldots ,r_{n})$ where $r_1, \ldots ,r_{l-1},r_{l+1}, \ldots ,r_n$ are single rows of $0$'s, and $r_l=\{s\}$ for a tuple  $s:A_i\mapsto i$, where $R_l=\{A_1, \ldots ,A_m\}$. It suffices then to chase $d$ by $\Sigma$ with rules (i,ii), and show that the resulting database $d''$ satisfies *. Since this is analogous to the previous case, we  omit  the proof here.
\end{proof}

We obtain the following corollaries. For the first corollary, note that in the last two cases of the previous proof none of the rules $\tax{5},\uax{3},\uax{4},\uax{5}$ are applied. The second corollary follows directly from the previous theorem which shows that the same axiomatization  characterizes both implication problems.
\begin{corollary}\label{corollary:IA+IND}
The axiomatization $\{\tax{1},\tax{2},\tax{3},\tax{4}\}\cup \{\uax{1},\uax{2}\}\cup\mathfrak{B}$ is sound and complete for the unrestricted and finite implication problems of DIA+IND.
\end{corollary}

\begin{corollary}\label{cor:coincide}
The finite and unrestricted implication problems of IND+IA coincide.
\end{corollary}

\subsection{Armstrong Databases}

Furthermore, it is straightforward to construct an Armstrong database based on all counterexample constructions. Given a class $\mathcal{C}$ of dependencies, an \emph{Armstrong database}  for a set $\Sigma$ of $\mathcal{C}$-dependencies is a relation that satisfies all $\mathcal{C}$-dependencies implied by $\Sigma$ and does not satisfy any $\mathcal{C}$-dependency not implied by $\Sigma$ \cite{DBLP:journals/jacm/Fagin82}. Armstrong databases are perfect sample databases as they reduce the implication problem for checking whether an arbitrary $\mathcal{C}$-dependency $\varphi$ is implied by a fixed set $\Sigma$ of $\mathcal{C}$-dependencies to checking whether $\varphi$ is satisfied in a $\mathcal{C}$-Armstrong database for $\Sigma$. This concept is useful for sample-based schema designs of databases and helps with acquisition of data dependencies that are meaningful for a given application domain \cite{DBLP:journals/is/LangeveldtL10,DBLP:journals/jcss/MannilaR86}. We say that some database is an Armstrong database with respect to finite implication if we replace above "implied" by "finitely implied". A uni-relational Armstrong database is called an \emph{Armstrong relation}.

\begin{theorem}
Let $\Sigma$ be a finite set of IA+IND. Then $\Sigma$ has a finite Armstrong database.
\end{theorem}
\begin{proof}
Without loss of generality we may consider only the uni-relational case.
We need to construct a finite relation $r$ which satisfies $\Sigma$ and falsifies any $\sigma$ not implied by $\Sigma$. Let $r_{\sigma}$ be a finite relation that satisfies $\Sigma$ and falsifies $\sigma$. We define an Armstrong relation $r$ as the set of tuples $t$ constructed as follows. First, from each $r_\sigma$ select some $t_\sigma$. Then $t$ maps $A$ to $(t_{\sigma_1}(A),\ldots ,t_{\sigma_n}(A))$ where $\sigma_1, \ldots ,\sigma_n$ is some enumeration of all IAs and INDs not implied by $\Si$. It is straightforward to verify that $r$ satisfies $\Sigma$ and falsifies each $\sigma_i$.
\end{proof}
This simple method can be extended to other dependency classes as well. However, in later sections we demonstrate how to combine Armstrong and counterexample constructions which has the effect of producing smaller Armstrong relations.

\subsection{Simplifying Implication}

Another consequence of Theorem \ref{theorem:ind+ia} is that the implication problem for UIND+CA by IND+IA can be determined by considering only interaction between UINDs and CAs.
For a set of dependencies $\Si$, define $\Sigma_{\rm CA}:=\{A\boto A\mid R:AX\boto AY\in \Si\}$ and  $\Sigma_{\rm UIND}:=\{R[A_i]\sub R'[B_i]\mid R[A_1 \ldots A_n]\sub R'[B_1\ldots B_n]\in \Si, i=1, \ldots ,n\}$. The following theorem now formulates this idea.
\begin{theorem}\label{theorem:uind+uca}
Let $\Sigma$ be a set of INDs and IAs, and let $\sigma$ be a UIND or a CA. The following are equivalent:
\begin{enumerate}[label=\emph{(\arabic*)}]
\item $\Sigma\models \sigma$,
\item $ \Sigma_{\rm UIND}\cup \Sigma_{\rm CA} \models  \sigma$,
\item $\sigma$ is derivable from $\Sigma_{\rm UIND}\cup \Sigma_{\rm CA}$ by $\iax{1},\iax{2},\uax{3},\uax{4}$.
\end{enumerate}
\end{theorem}
\begin{proof}
It is clear that $(3)\Rightarrow (2)\Rightarrow (1)$. We show that $(1)\Rightarrow (3)$. By Theorem \ref{theorem:ind+ia}, there is a deduction $(\sigma_1, \ldots ,\sigma_m)$ from $\Sigma$ by  $\mathfrak{I}\cup\mathfrak{B}\cup\mathfrak{C}$ such that $\sigma_m=\sigma$. It is a straightforward induction to show that for all $i=1, \ldots ,m$:
\begin{itemize}
\item If $\sigma_i$ is   $R:A \boto A$, then $\sigma_i$ satisfies (3).
\item If $\sigma_i$ is   $R[A_1\ldots A_n]\sub R'[B_1\ldots B_n]$, then $\sigma_j:=R[A_j]\sub R'[B_j]$ satisfies (3), for $j=1, \ldots ,n$.
\end{itemize}
It is worth noting that every application of $\uax{5}$, where $\sigma(R:A\mapsto B)$ is a UIND or CA, can be simulated by $\iax{2},\uax{3},\uax{4}$. All the other cases are straightforward and left to the reader.
\end{proof}

\section{Independence Atoms and Functional Dependencies}\label{sect:FD+IAgen}

In this section we consider the interaction between functional dependencies and independence atoms. Already keys and IAs combined form a somewhat intricate class: Their finite and unrestricted implication problems differ and the former lacks a finite axiomatization \cite{DBLP:conf/cikm/HannulaKL14,HannulaKL16}. In Section \ref{sect:FD+IA} we will extend these results to the classes FD+IA and $2$FD+UIA. However, the interaction between unary FDs and IAs is less involved. In Section \ref{sect:UFD+IA} we will show that for UFD+IA unrestricted and finite implication coincide and the axiomatization $\mathfrak{A}^*$, depicted in Table \ref{tab-rules}, forms a sound and complete axiomatization.

For notational clarity we restrict attention to the uni-relational case from now on. That is, we  consider only those cases where databases consist of a single relation.

\subsection{Implication problem for FDs and IAs}\label{sect:FD+IA}

The following theorem enables us to separate the finite and unrestricted implication problems for FD+IA as well as for FD+UIA.
\begin{theorem}[\cite{HannulaKL14}]
The unrestricted and finite implication problems for keys and UIAs differ.
\end{theorem}
This theorem was proved by showing that $\Si\fmodels\si$ and $\Si \not\models \si$, for $\Si:=\{A\boto B,C\boto D, BC\to AD,AD\to BC\}$ and $\si:=AB\to CD$. Next we show how this counterexample  generalizes to a non-axiomatizability result for the finite implication problem of FD+IA. For $n\geq 2$, define $R_n:=\{A_i,B_i:i=1, \ldots ,n\}$ and $\Si_n:=\{A_i\boto B_i:i=1, \ldots ,n\}\cup \{A_{\suc(i)}B_i\to R_n:i=1, \ldots ,n\} $ where $\suc(n)=1$ and $\suc(i)=i+1$, for $i<n$. We say that $\si$ follows from $\Si$ by $k$-ary (finite) implication, written $\Si\models^k \si$ ($\Si\fmodels^k \si$),  if $\Si'\models \si$ for some $\Si'\sub\Si$ of size at most $k$. We say that an inference rule of the form $\si_1, \ldots ,\si_k\Rightarrow \si$ is $k$-ary. An axiomatization is called $k$-ary if it consists of at most $k$-ary rules.
In  \cite{HannulaKL14} it was shown that, for $n\geq 2$,
\begin{enumerate}[label=\emph{(\arabic*)}]
\item $\Si_n\fmodels A_1B_1 \to R_n$;
\item $\Si_n\fmodels^{2n-1} X\boto Y$ iff $X,Y$ are disjoint and such that $XY=A_iB_i$, $X=\emptyset$, or $Y=\emptyset$;
\item given  $\Si'\sub\Si_n$ of size $2n-1$ and $X\sub R$ such that $A_{\suc(i)}B_i \not\sub X$ for all $i=1, \ldots ,n$, one finds  a relation $r$ satisfying $\Si'$ and  tuples $t,t'\in r$ such that $t(A)=t'(A)$ iff $A\in X$.
\end{enumerate}
It follows from (3) that
\begin{enumerate}[label=\emph{(\arabic*)}]
  \setcounter{enumi}{3}
\item $\Si_n\fmodels^{2n-1} X\to Y$ iff $Y \sub X$ or $A_{\suc(i)}B_i \sub X$ for some $i=1, \ldots ,n$.
\end{enumerate}
Note that all  FDs and IAs described in (2) and (4) follow from $\Si_n$ by unary finite implication. Consequently, the set of FDs and IAs described in (2) and (4) is closed under $(2n-1)$-ary finite implication. Hence, the finite implication problem for FDs and IAs cannot have any $(2n-1)$-ary sound and complete  axiomatization. Since this holds for arbitrary $n\geq 2$, and since all IAs in $\Si_n$ are unary, we obtain the following theorem.
\begin{theorem}\label{thm:noaxioms}
The finite implication problem for FD+IA (2FD+UIA) is not finitely axiomatizable.
\end{theorem}
To the best of our knowledge, decidability is open for both FD+IA and FD+UIA with respect to their finite and unrestricted implication problems. It is worth noting here that the
unrestricted (finite) implication problem for FD+UIA is as hard as that for FD+IA. For this, we demonstrate a simple reduction from the latter to the former. Let $\Si\cup\{\si\}$ be a set of FDs and IAs, and let $\Si'$ denote the set of FDs and IAs where each IA of the form $X\boto Y$ is replaced with  dependencies from $\{A\boto B,X\to A,A\to X,Y\to B,B\to Y\}$ where $A$ and $B$ are fresh attributes. If $\si$ is an FD, then $\Si$ (finitely) implies  $\si$ iff $\Si'$ (finitely) implies $ \si$. Also, if $\si$ is of the form $X\boto Y$, then we have  $\Si\models \si$ iff $\Si''\models \si'$, where
\[
\Si'':=\Si'\cup \{X\to A,A\to X,Y\to B,B\to Y\},
\]
$\si':=A\boto B$, and $A$ and $B$ are fresh attributes.


\subsection{Implication problem for UFDs and IAs}\label{sect:UFD+IA}
Next we turn to the class UFD+IA. We show that $\mathfrak{A}^*$ (see Table \ref{tab-rules}) forms a sound and complete axiomatization for UFD+IA in both the finite and unrestricted cases. Hence, compared to UIAs and FDs, the interaction between IAs and UFDs is relatively tame. Combined, however, these two may entail new restrictions to column sizes. For instance, in the finite $A\to B_1$, $A\to B_2$, and $B_1\boto B_2$  imply $|r(B_1)|\cdot |r(B_2)|\leq |r(A)|$.

We also show that the class UFD+IA admits Armstrong relations.  We start with the following auxiliary lemma which shows existence of finite Armstrong relations in the absence of CAs. 

\begin{lemma}\label{lemma:IAhelp}
Let $\Sigma$ be a finite set of UFDs and IAs that is closed under $\mathfrak{A}^*$ and such that it contains no CAs. Then $\Sigma$ has a finite Armstrong relation. %
\end{lemma}
\begin{proof}
For an attribute $A$ define $A^+$ as the set of  attributes $B$ such that $A\to B \in \Sigma$. Enumerating the underlying relation schema of $\Sigma$ as $R=\{A_2, \ldots ,A_n\}$ we first let $r_0$ to consist of two constant tuples $t_0$ and $t_1$ that respectively map all attributes to $0$ and $1$, and tuples $t_i$ that map attributes in $A_i^+$ to $0$ and those in $R\setminus A_i$ to $i$. By transitivity we observe that the condition (i) is satisfied by $r_0$. Next we extend $r_0$ to a relation satisfying both conditions. First we chase $r_0$ by all IAs in $\Si$ with the following rule where $b$ is a new symbol.

\begin{itemize}
\item Assume that $X\boto  Y\in \Si$ and $t,t'\in r$ are such that for no $t''\in r$, $t''(X)=t(X)$ and $t''(Y)=t'(Y)$. Then extend $r$ with $\tn$ such that $\tn(X)=t(X)$, $\tn(Y)=t'(Y)$, and $\tn(R\setminus XY)=b$.
\end{itemize}
Note that $X$ and $Y$ are disjoint for all $ X\boto Y\in \Si$, and thus the rule is well defined. Assuming $X$ and $Y$ share an attribute $A$ generates a deduction of $\emptyset \to A$ using decomposition, reflexivity, and constancy. This, in turn, contradicts our assumption.

Since each $\tn$ assigns attributes  only to $n+1$ different values, the chase procedure terminates and generates a unique $r_1$ satisfying all IAs in $\Si$. 
Let $r$ be obtained from $r_1$ by replacing each $t\in r$ with the tuple $s_t$ 
which maps $A$ to the tuple $t( A^+)$. We claim that $r$ is a finite Armstrong relation. First note that $r$ is of the size of $r_1$ and hence  finite. In the following we consider the cases of UFDs and IAs.

 Let us show first that $r$ satisfies an arbitrary UFD $A\to B$ from $\Si$. Let $s_t,s_{t'}$ be tuples in $r$ that agree on $A$. Then $t$ and $t'$ agree on $A^+$, and thus on $B^+$, from which it follows that $s_t$ and $s_{t'}$ agree on $B$.

Suppose then $A_i\to B$ is not in $\Si$. Then  $t_0$ and $t_i$ map $A_i$ to $0$ but $B$ to $0$ and $i$, respectively. By definition, $s_{t_0}$ and $s_{t_i}$ agree on $A_i$ and disagree on $B$, and thus $r$ does not satisfy $A_i \to B$. Furthermore, note that any $\emptyset \to A$ is falsified by $s_{t_0}$ and $s_{t_1}$.


 Let us then show  that $r$ satisfies an arbitrary  $X\boto Y$ in  $\Si$. Let $s_t,s_{t'}\in r$, and let $t''\in r$ be such that it agrees with $t$ on $X$ and with $t'$ on $Y$. We claim that $s_{t''}$ agrees with $s_t$ on $X$ and with $s_{t'}$ on $Y$. Let $A\in X$; we show that $s_t(A)=s_{t''}(A)$. 
 By symmetry and composition IA components are closer under UFDs which means that we have $A^+\sub X$. Hence, $t$ agrees with $t''$ on $ A^+$, and therefore $s_t$ agrees with $s_{t''}$ on $A$. Consequently, $s_t(X)=s_{t''}(X)$, and analogously we obtain that $s_{t'}(Y)=s_{t''}(Y)$. This shows that $r$ satisfies $X\boto Y$.

To show  that $r$ satisfies no  $X\boto Y\notin\Si$, it suffices to show that  $ X\boto Y\in \Si$ if for some $t\in r$ we have $t(X)=0$ and $t(Y)=1$. Proving this is a straightforward induction on the chase construction. For $t_0$ and $t_1$ it suffices to employ $\tax{1}$. For the induction step it suffices to use Lemma \ref{lemma:help3} and rules $\tax{2},\tax{3}$ since  no new tuple $\tn$, obtained by an application of the chase rule to $X\boto Y$, maps attributes in $R\setminus XY$ to $0$ or $1$.
\end{proof}

\begin{theorem}
\label{thm:cc+uind}
The axiomatization $\mathfrak{A}^*$ is sound and complete for the unrestricted and finite implication problems of UFD+IA. Furthermore, any finite set of UFDs and IAs has a finite Armstrong relation.
\end{theorem}
\begin{proof}
By Theorem \ref{soundness} the axiomatization is sound.  To show completeness and existence of Armstrong relations, it suffices to show that there is a finite Armstrong relation for any finite set $\Sigma$ of UFDs and IAs that is closed under $\mathfrak{A}^*$.
Let $R'=R\setminus C$ for $C:=\{A\in R\mid  \emptyset \to A\in \Sigma\}$. First we note that $\Sigma \upharpoonright R'$ is derivable from $\Sigma$. For an IA $X\boto Y \in \Sigma$,  we may derive $X \setminus C\boto Y\setminus C$ by symmetry and decomposition. For a UFD $A\to B \in \Sigma$ we consider the only non-trivial case: $B$ intersects with $R'$ but $A$ does not. In this case $\emptyset \to A\in \Sigma$ and thus we obtain by transitivity that $\emptyset \to B\in \Sigma$, which contradicts our assumption.

Let $\Sigma'$ be the closure of $\Sigma \upharpoonright R'$ under $\mathfrak{A}^*$. Having concluded that $\Sigma'$ is derivable from $\Sigma$, we observe that it cannot contain any CAs for otherwise we would derive from $\Sigma$ some $\emptyset \to A$ where $A\in R'$, and thus $A\in C$ which is a contradiction. Hence, we may apply the previous lemma to obtain a finite relation $r_0$ satisfying exactly those UFDs and IAs that belong to $\Sigma'$. Define $r:={}^{C}\{0\}\times r_0$. We claim $r$ that is an Armstrong relation for $\Sigma$.

 Assume first $X\boto Y \in\Sigma$, and consider $X\setminus C\boto Y\setminus C \in \Sigma'$. Since $r_0$ satisfies $X\setminus C\boto Y\setminus C$, it follows that $r$ satisfies $X\boto Y$. Assume then that $U\to B\in \Si$. If $U\in C$ or $U$ is the empty set, then $B\in C$ by transitivity from which it follows that $r$ satisfies $U\to B$ trivially. The same happens if $B\in C$. If $U$ is non-empty and $U,B\not\in C$, then $U\to B$ is  satisfied by $r_0$ and thus by  $r$, too.

 Suppose then $X\boto Y \notin\Sigma$. 
 By  symmetry and composition we note that $X\setminus C\boto Y\setminus C\notin\Sigma'$. Hence,  $r_0$ does not satisfy $ X\setminus C \boto Y\setminus C$ and thus $r$ does not satisfy $ X\boto Y$. Suppose then $U\to B\notin \Sigma$. If $U$ is the empty set, then $B$ is not in $C$, and thus $U\to B$ does not hold in $r_0$ nor in $r$. If $U$ is an attribute, then again by reflexivity and transitivity $B$ cannot be in $C$. If $U$ is in $C$, then $U\to B$ is false in $r$ because $U$ is a constant and $B$ is not. If $U$ is not in $C$, then $U\to B$ is false in $r_0$ and consequently in $r$ as well. This concludes the proof.
   \end{proof}

As the same axiomatization characterizes both finite and unrestricted implication, we obtain the following corollary.
\begin{corollary}\label{cor:cc+uind}
The finite and unrestricted implication problems coincide for UFD+IA.
\end{corollary}

\section{Independence atoms, unary functional dependencies and unary inclusion dependencies}\label{sect:UFD+UIND+IA}

In this section we consider the combined class of FDs, INDs, and IAs. In the previous section we noticed that the finite implication problem for binary FDs and unary IAs is not finitely axiomatizable. On the other hand, both the finite and unrestricted implication problems for unary FDs and binary INDs are undecidable \cite{mitchell83}. Hence, in this section we will restrict our interest to unary FDs and unary INDs, a class for which finite and unrestricted implication problems already deviate \cite{CosmadakisKV90}. It turns out that the combination of unary functional dependencies, unary inclusion dependencies, and arbitrary independence atoms can be axiomatized with respect to finite and unrestricted implication. However, in the finite case the axiomatization is infinite as one needs to add so-called cycle rules for UFDs and UINDs.

\subsection{Unrestricted implication}

Let us first consider unrestricted implication for the class UFD+UIND+IA. Based on Section \ref{sect:UFD+IA} and \cite{CosmadakisKV90} we show that $\mathfrak{A}^*\cup \{\iax{1},\iax{2},\uax{3},\uax{4}\}$ forms a sound and complete axiomatization. An axiomatization for the unrestricted implication problem  of UINDs and embedded implicational dependencies was shown in \cite{CosmadakisKV90}. In the uni-relational case, an
\emph{embedded implicational dependency} (EID) is a first-order sentence  $\forall \tuple x \exists \tuple y(\phi(\tuple x)\to  \psi(\tuple x, \tuple y))$  over a relational vocabulary $\{R\}$ where
\begin{itemize}
\item $\phi(\tuple x)$  is a non-empty finite conjunction of relational atoms and the
variables occurring in this conjunction are exactly all the variables listed in $\tuple x$;
\item $\psi(\tuple x,\tuple y)$ is a non-empty finite conjunction of relational atoms and
equality atoms and the variables occurring in this conjunction are exactly all the variables listed in $\tuple y$ and some of the variables listed in $\tuple x$;
\item each variable in $\tuple x\tuple y$ associates with a single relation position, and for each occurrence of an equality atom $x=y$, $x$ and $y$ associate with the same relation position.
\end{itemize}
Note that EIDs include all FDs and IAs but exclude all non-trivial INDs.
The following presentation of Theorem \ref{thm:cosma} is a reformulation from  \cite{CosmadakisKV90}.

\begin{definition}[\cite{CosmadakisKV90}]\label{span} For any set $\Delta$ of EIDs and set $\Si$ of UINDs, we define the set $Y$ called the singlevalued span of $\Delta$ and $\Si$ to be the minimum set of attributes $Y$ that satisfies the two conditions:
\begin{enumerate}[label=\emph{(\arabic*)}]
\item if 
 $\Delta\cup\{Y \boto Y\}
\models A\boto A$, then add $A$ to $Y$,
\item if attribute $B$ is in $Y$ and $A\sub  B$ is in $\Si$, then add $A$ to $Y$.
\end{enumerate}
\end{definition}

\begin{definition}[\cite{CosmadakisKV90}] For any set $A$ of EIDs, any set $Z$ of UINDs, and $Y$ the singlevalued span of $\Delta$ and $\Si$, we define the sets $\Delta''$ and $\Si''$, called the unrestricted extensions of $\Delta$ and $\Si$, by $\Delta'' = \Delta \cup \{ Y\boto Y\}$ and $\Sigma'' = \Sigma \cup \{A \sub B : B \sub A \textrm{ in }\Si, B \textrm{ in } Y\}$.
\end{definition}
\begin{theorem}[\cite{CosmadakisKV90}]\label{thm:cosma} Let $\Delta$ be a set of EIDs, $\Si$ set of UINDs, $Y$ the singlevalued span, and $\Delta'',\Si''$ the unrestricted extensions of $\Delta, \Si$. For any ED $\delta$ and any UIND $\sigma$, we have
\begin{itemize}
\item $\Delta\cup \Si\models \sigma \Leftrightarrow \Sigma''\models \sigma$,
\item $\Delta\cup \Si\models \delta \Leftrightarrow \Delta''\models \delta$.
\end{itemize}

\end{theorem}
By Theorems \ref{soundness}, \ref{thm:cc+uind},  \ref{thm:cosma}, and since  $\iax{1},\iax{2}$ form a complete axiomatization for UINDs, we  obtain the following theorem. 
\begin{theorem}\label{thm:unraxioms}
The axiomatization  $\mathfrak{A}^*\cup \{\iax{1},\iax{2},\uax{3},\uax{4}\}$ is sound and complete for the unrestricted implication problem of UFD+UIND+IA.
\end{theorem}
\subsection{Finite Implication}
For the finite implication problem of UFD+UIND+IA we obtain a complete axiomatization by extending the axioms in Theorem \ref{thm:unraxioms} with the so-called cycle rules \cite{CosmadakisKV90} (see Table \ref{tab-rules3}) and  by removing $\uax{3},\uax{4}$ which become redundant.
The idea of the completeness proof is to combine the chase-based approach of the proof of Theorem \ref{thm:cc+uind} with a graph-theoretic approach from \cite{CosmadakisKV90} that was used to prove a complete axiomatization for the finite implication problem of UFD+UIND. In this graphical approach a given constraint set $\Sigma$ is first closed under the inference rules and then interpreted as a graph with edges of two different colors.

\begin{table}

\[\fbox{\scalebox{.96}[.96]{$\begin{array}{c}

\cfrac{A_1 \to A_2\hspace{.3cm} A_2\supseteq A_3\hspace{.3cm} \ldots \hspace{.3cm} A_{2n-1}\to A_{2n} \quad A_{2n}\supseteq A_1
}{A_1 \leftarrow A_2\hspace{.3cm} A_2\sub A_3\quad \ldots \hspace{.3cm} A_{2n-1}\leftarrow A_{2n} \hspace{.3cm} A_{2n}\sub A_1 }

\\
\text{(cycle rule for $n$, $\mathcal{C}_n$)}

\end{array}$}}\]
\caption{Cycle rules for finite implication \label{tab-rules4}}
\vspace{-2mm}
\end{table}

\begin{definition}[\cite{CosmadakisKV90}]\label{def:multigraph}
For each set $\Si$ of UINDs and UFDs over $R$, let $G(\Sigma)$ be the multigraph that consists of nodes $R$, red directed edges $(A,B)$, for $A\to B \in \Si$, and black directed edges $(A,B)$, for $B\sub A\in \Si$. If $G(\Si)$ has red (black) directed edges from $A$ to $B$ and vice versa, then these edges are replaced with an undirected edge between $A$ and $B$.
\end{definition}
Given a multigraph $G(\Si)$, we then topologically sort its strongly connected components which form a directed acyclic graph \cite{Kahn62}. That is, each component is assigned a unique \emph{scc-number}, greater than the scc-numbers of all its descendants. For an attribute $A$, denote by $\scc(A)$ the scc-number of the component node $A$ belongs to. Note that $\scc(A)\geq \scc(B)$ if $(A,B)$ is an edge in $G(\Si)$.  Denote also by $\scc_i$ the set of attributes $A$ with $\scc(A)=i$, and let $\scc_{\leq i}:=\bigcup_{j\leq i} \scc_j$ and define $\scc_{ \geq i}$, $\scc_{<i}$, and $\scc_{>i}$  analogously. The following lemma is a simple consequence of the definition.\footnote{Lemma \ref{lem:cosma0} is a reformulation of Lemma 4.2.  in \cite{CosmadakisKV90} where the same claim is proved for a set of FDs and UINDs that is closed under $\{\fax{1},\fax{2},\fax{3},\iax{1},\iax{2}\}\cup\{\C_k:k\in \N\}$. We may omit $\fax{3}$ here since, when restricting attention to UFDs, $\fax{3}$ is not needed in the proof.}

\begin{lemma}[\cite{CosmadakisKV90}]\label{lem:cosma0}
 Let $\Si$ be a set of UFDs and UINDs that is closed under \[\{\fax{1},\fax{2},\iax{1},\iax{2}\}\cup\{\C_k:k\in \N\}.\] Then every node in $G(\Si)$ has a red and a black self-loop. The red (black) subgraph of $G( \Si)$ is transitively closed. The subgraphs induced by the strongly connected components of $G(\Si)$ are undirected. In each strongly connected component, the red (black) subset of undirected edges forms a collection of node-disjoint cliques. Note that the red and black partitions of nodes could be different.
\end{lemma}
Using this lemma we prove the following result which essentially shows how to construct a counterexample model for the finite implication problem of UFD+UIND+IA.  The construction is somewhat intricate, mainly because of the requirements for balancing conditions (i) and (iii) against one another.

\begin{lemma}\label{lem:help1}
Let  $\Si=\Sig{UFD}\cup\Sig{UIND}\cup\Sig{IA}$ be a set of UFDs, UINDs, and IAs that is closed under $\mathfrak{A}^*\cup\{\iax{1},\iax{2}\}\cup\{\C_k:k\in \N\}$ and such that it contains no CAs.
Let $0, \ldots ,n$ be some $\scc$-numbering of $G(\Sig{UFD}\cup\Sig{UIND})$, and let $M_1, \ldots ,M_n$ be a sequence of positive integers. Then there exists a finite relation $r$ and a sequence of positive integers $N_0, \ldots ,N_n$ such that $N_{i}\geq N_{i-1}+M_i$ and
\begin{enumerate}[label=\emph{(\roman*)}]
\item $r\models A\to B $ iff $A\to B \in \Si$,
\item $r\models X\boto Y$ iff $X\boto Y \in \Si$, and
\item  $|r(A)| =N_i$ for $A\in \scc_i$ where $i\geq 1$.
\end{enumerate}
\end{lemma}
\begin{proof}
By Theorem \ref{thm:cc+uind} there is a finite relation $r$ satisfying (i-ii). We show how to extend  $r$ to a finite relation satisfying (i-iii). 
We construct inductively, for each $i=0, \ldots ,n$,  a relation $r_i$  such that it
  satisfies (i-iii) with the alleviation that (iii) holds over $A\in \scc_{\leq i}$.

Without loss of generality $\Si$ contains some CAs in which case $\scc_0$ consists of all derivable constants and we set  $r_0= r$ and $N_0=1$. 

 For the induction step, assuming   $r_{i-1}$ satisfies the induction claim we show how to define a relation $r_i$ satisfying the induction claim.
First observe that we may rename the values of $r_{i-1}$ 
 such that for all attributes $A$, the set $r_{i-1}(A)$ is the initial segment $\{1, \ldots ,|r_{i-1}(A)|\}$ of $\N$. 
  Note that satisfaction of typed dependencies is invariant under such renaming. Furthermore, 
   this renaming can be done in such a way that all columns in the same maximal red clique are identical.
   In what follows we now construct $r_i$.

Let $A^+$ be the set of all $B$ such that $A\to B \in \Si$. 
  First we define an auxiliary relation $r$ whose tuples are obtained from those of $r_{i-1}$
 by replacing each value $|r_{i-1}(A)|$ of $A\in \scc_i$ with any value from \[\{|r_{i-1}(A)| ,\ldots ,N_i\},\]
 where $N_i$ is the maximun of  $\max\{|r_{i-1}(A)|: A\in \scc_i\}$ and $N_{i-1}+M_{i}$.
 Furthermore, we place the restriction that the same value must be picked for each attribute from the same maximal red clique.
 That is, we replace all maximal attribute values at the $i$th $\scc$ level with some new value consistently with other attributes in the same maximal red clique, and the new maximal value is at least $N_{i-1}+ M_i$. 
We then define $r_i:=\{ s_t\mid t\in r\}$ where $s_t$ maps each attribute $A$ 
 to the tuple $t(A^+)$.

Let us consider the items in the following.
\begin{enumerate}[label=\emph{(\roman*)}]
\item Let $A\to B\in \Si$. 
Let $s_t(A)=s_{t'}(A)$ for $s_t$ and $s_{t'}$ from $r_i$. By definition $t(A^+)=t'(A^+)$ which entails by $B^+\sub A^+$ that $t(B^+)=t'(B^+)$, and hence $s_t(B)=s_{t'}(B)$.

Suppose $A\to B\notin \Si$. By induction assumption we find $t,t'$ from $r_{i-1}$ which agree on $A$ but disagree on $B$. Since $r_{i-1}$ satisfies all UFDs of $\Si$ we note that they agree on $A^+$ but disagree on $B^+$. This means $s_t$ and $s_{t'}$ witness that $A \to B$ is false in $r_i$.


\item Suppose $X\boto Y\in \Si$. Without loss of generality $X$ and $Y$ are disjoint.
 Consider tuples $s_{t},s_{t'}$ in $r_i$. By induction assumption we find $t''$ from $r_{i-1}$ which agrees with $t$ on $X$ and with  $t'$  on $Y$ when restricted to only those attributes $A\in XY\cap \scc_i$ with values  at most $|r_{i-1}(A)|$. Copying the remaining attribute values of $X$ and $Y$ from $t$ and $t'$ we obtain a modification of $t''$ which agrees with $t$ on $X$ and with  $t'$  on $Y$. Since $X$ and $Y$ are closed under the UFDs of $\Si$, $s_{t''}$ which agrees with $s_{t}$ on $X$ and with $s_{t'}$ on $Y$.  

 Suppose $X\boto Y\notin \Si$. By induction assumption $X\boto Y$ is not satisfied in $r_{i-1}$, that is, we find $t$ and $t'$ from $r_{i-1}$ such that  no $t''$ from $r_{i-1}$ agrees both on $X$ with $t$ and on $Y$ with $t'$. Considering $r_i$, finding a tuple $s_{t''}$ which agrees on $X$ with $s_t$ and on $Y$ with $s_{t'}$ leads to a contradiction with the previous statement. Thus $X\boto Y$ remains not satisfied in $r_i$.

\item 
Since $r_{i-1}$ is closed under the UFDs of $\Si$, it follows by induction assumption that $|r_i(A)|=|r_{i-1}(A)| =N_i$ for $A \in \scc_{<i}$. To show that this holds also for $A\in \scc_i$, first note that 
 the restrictions of $r_{i-1}$ to $A$ and $A^+$ are of the same size. Furthermore,  $r$ is obtained by replacing only the values $|r_{i-1}(A)|$ with new values simultaneously with all other attributes in the same maximal red clique. It follows that $r(A^+)$, and thus $r_i(A)$, is of size $N_i$.
\end{enumerate}
\end{proof}

Using the above counterexample construction we may now prove the completeness theorem for finite implication of UFD+UIND+IA. As a consequence of our proof technique the finite implication of UFD+UIND+IA enjoys finite Armstrong relations.

\begin{theorem}
\label{thm:finaxioms}
The axiomatization $\mathfrak{A}^*\cup \{\iax{1},\iax{2}\}\cup\{\cax{n}: n\in \N\}$ is sound and complete for the finite implication problem of UFD+UIND+IA. Furthermore, any finite set of UFDs, UINDs, and IAs has a finite Armstrong relation with respect to finite implication.
\end{theorem}

\begin{proof}
By Theorem \ref{soundness} and by soundness of the cycle rule the axiomatization is sound.  To show completeness and existence of Armstrong relations, it suffices to show that there is a finite Armstrong relation for any finite set $\Sigma$  of UFDs, UINDs, and IAs that is closed under the axiomatization.  Let $r_0$ be a finite relation obtained by the previous lemma. By condition (iii) $r_0$ satisfies also all UINDs of $\Sigma$ and is thus a model of $\Sigma$. Without loss of generality we have $r_0(A)=\{1, \ldots ,N_i\}$ for each attribute $A$ at level $i$ of the $\scc$ numbering of $G(\Sig{UIND}\cup\Sig{UFD})$. We may assume $N_i\geq  N_{i-1}+M_i$ where $M_i$ is the number of maximal black cliques at levels at most $i$. Consider the graph $G(\Sig{UIND})$, i.e., the subgraph of $G(\Sig{UIND}\cup\Sig{UFD})$ obtained by removing all the red edges. We may define an scc-numbering $\scc'$ for $G(\Sig{UIND})$ such that $\scc'(A) \leq \scc'(B)$ implies $\scc(A)\leq \scc(B)$.

We now modify $r_0$ as follows. If  there is a black edge from an attribute $A$ to another attribute $B$, then on $A$ replace  $N_{i-1}+\scc'(B)$ with $N+\scc'(B)$ where $i=\scc(A)$. By transitivity of UINDs and since $M_i$ was chosen large enough this is operation is well defined. We claim that the relation $r$ obtained by these modifications is an Armstrong relation for $\Si$. Again, satisfaction of typed dependencies is invariant under renaming of attribute values, and thus $r$ is Armstrong with respect to all UFD and UIND consequences. Assume then $C\sub D \notin \Sigma$. If $\scc(D) < \scc(C)$, then this UIND is false because the cardinality of $r(C)$ is strictly greater than that of $r(D)$. Otherwise, there is no black egde from $D$ to $C$ but by reflexivity $C$ has a black self-loop. Thus $N+\scc'(C)$ appears in $r(C)$ but not in $r(D)$. Suppose then $C\sub D \in \Si$, and let $a\in r(C)$. If $a$ is $N+\scc'(B)$ for some attribute $B$, then there is a black edge from $C$ to $B$ and by transitivity from $D$ to $B$ which implies that $N+\scc'(B)$ appears in $r(D)$ as well. Otherwise, $a$ is not greater than $N_{i}$ for $\scc(C)=i$. Now either $D$ and $C$ are in the same maximal black clique or $\scc(D)>i$.  In the first case $a$ must appear in $r(D)$, and in the second case $r(D)$ contains all positive integers from $1$ to $N_{i}$. We conclude that $r$ is also Armstrong with respect to UIND consequences.
\end{proof}


\section{Polynomial-Time Conditions for Non-Interaction}\label{sect:poly}

Naturally, the implication problems for a combined class are more difficult than the corresponding implication problems for the individual classes. While finite and unrestricted implication problems for FDs coincide and are PTIME-complete and finite and unrestricted implication problems for INDs coincide and are PSPACE-complete, finite and unrestricted implication problems for FDs and INDs deviate and are undecidable. However, the literature has brought forward tractable conditions that are sufficient for the non-interaction of these classes~\cite{DBLP:journals/ipl/LeveneL99,DBLP:journals/tcs/LeveneL01}. That is, whenever these conditions are met, then the implication for FDs and INDs by a given set of FDs and INDs can be determined by the restriction of the given set to FDs and INDs alone, respectively. Hence, non-interaction is a desirable property. In this section we examine the frontiers for tractable reasoning about the classes of FD+IA and IND+IA, respectively, in both the finite and unrestricted cases. The idea is to establish sufficient criteria for the non-interaction between IAs and FDs, and also between IAs and INDs. There is a trade-off between the simplicity and generality of such criteria. While simple criteria may be easier to apply, more general criteria allow us to establish non-interaction in more cases. Our focus here is on generality, and the criteria are driven by the corresponding inference rules. We define non-interaction between two classes as follows.
\begin{definition}
Let $\Si_0$ and $\Si_1$ be two sets of dependencies from classes $\C_0$ and $\C_1$, respectively. We say that $\Si_0,\Si_1$ have \emph{no interaction with respect to unrestricted (finite) implication} if
\begin{itemize}
\item for $\si$ from $\C_0$, $\si$ is (finitely) implied by $\Si_0$ iff $\si$ is (finitely) implied by $\Si_0\cup \Si_1$.
\item for $\si$ from $\C_1$, $\si$ is (finitely) implied by $\Si_1$ iff $\si$ is (finitely) implied by $\Si_0\cup \Si_1$.
\end{itemize}
\end{definition}
Let us now define two syntactic criteria for describing non-interaction. We say that an IA $X\boto Y$ \emph{splits} an FD $U\to V$ (or an  IND $Z \sub  W$) if both $(X\setminus Y)\cap U$ and $(Y\setminus X)\cap U$ (or $ X\cap W $ and $ Y\cap  W $) are non-empty. The IA $X\boto Y$ \emph{intersects} $U\to V$ ($ Z\sub  W$) if $XY\cap U$ ($XY\cap  W$) is non-empty. Notice that both these concepts give rise to possible interaction between two different classes. We show that lacking splits implies non-interaction for IND+IA, and for FD+IA in the unrestricted case. Non-interaction for FD+IA in the finite is guaranteed by the stronger condition in terms of lacking intersections.

The proof for IND+IA is straightforward by using its complete axiomatization.
\begin{theorem}
Let $\Sig{IND}$ and $\Sig{IA}$ be respectively sets of INDs and IAs. If no IA in $\Sig{IA}$ splits any IND in $\Sig{IND}$, then $\Sig{IND}$ and  $\Sig{IA}$ have no interaction with respect to unrestricted (finite) implication.
\end{theorem}
\begin{proof}
Assume that $\si$ is implied by $\Sig{IND}\cup\Sig{IA}$ (recall that here finite and unrestricted implication  coincide). By Theorem \ref{theorem:ind+ia}, $\si$ can be deduced from $\Sig{IND}\cup\Sig{IA}$ by
$\mathfrak{I}\cup\mathfrak{B}\cup\mathfrak{C}$. Given the condition, no rules from $\mathfrak{C}$ can be applied in the deduction. Since only rules in $\mathfrak{B}$ (in $\mathfrak{I}$) produce fresh INDs (IAs), the claim follows.
\end{proof}
 The non-interaction results for FD+IA require more work.
For unrestricted implication the idea is to first apply the below polynomial-time algorithm which transforms an assumption set $\Si$ to an equivalent set $\Si^*$. The set $\Si^*$ is such that it has no interaction between FDs and IAs provided that none of its FDs split any IAs. Substituting $\Si$ for $\Si^*$ the same holds for finite implication.. These claims will be proven using a graphical version of the chase procedure.

Let us first consider the computation of $\Si^*$.  For a set of FDs $\Si$ we denote by $\textrm{Cl}(\Si,X)$ the closure set of all attributes $A$ for which $\Si\models X\to A$. This set can be computed in linear time by the Beeri-Bernstein algorithm \cite{BeeriB79}.  The non-interaction condition for unrestricted implication is now  formulated using $\Sig{IA}^*=\{X_1\boto Y_1, \ldots ,X_n\boto Y_n\}$ and $\Sig{FD}^*=\Sig{FD}\cup\{\emptyset \to Z\}$ where $Z,X_iY_i$ are computed using the following algorithm  that takes an FD set $\Sig{FD}$ and an IA set $\Sig{IA}=\{U_1\boto V_1, \ldots , U_n\boto V_n\}$ as an input.

\begin{algorithm}
			\caption{Algorithm for computing $Z,X_i,Y_i$}
			\label{algZ}
		\begin{algorithmic}[1]
\Require $\Si_{\rm FD}$ and $\Si_{\rm IA}=\{U_i\boto V_i\mid i=1, \ldots ,n\}$
 \Ensure $Z$ and $\Si^*_{\rm IA}=\{X_i\boto Y_i\mid i=1, \ldots ,n\}$
 \State \textbf{Initialize:}   $V\gets \emptyset, X_i \gets U_i,Y_i\gets V_i $
			\Repeat
			\State $Z\gets V$
			\For{$i=1, \ldots ,n$}
				\State  $X_i\gets \textrm{Cl}(\Sig{FD},X_iV)$
				\State  $Y_i\gets\textrm{Cl}(\Sig{FD},Y_iV)$
				\State $V\gets V\cup( X_i\cap Y_i)$
			\EndFor

			\Until{Z=V}

		\end{algorithmic}
	\end{algorithm}

From the construction we obtain that $\Sig{FD}^*\cup \Sig{IA}^*$ is equivalent to $\Sig{FD}\cup\Sig{IA}$ and that
\begin{enumerate}[label=\emph{(\arabic*)}]
\item for   $Z_1\boto Z_2\in \Sig{IA}^*$ and $i=1,2$, $\Sig{FD}^*\models Z_i\to X$ implies $X\sub Z_i$;
\item $\Sig{FD}^*\cup\Sig{IA}^*\models \emptyset \to A$ iff $ A \in Z  $.
\end{enumerate}
Recall that the closure set $\textrm{C}(\Sig{FD},X)$ can be computed in linear time by the Beeri-Bernstein algorithm. Now, at stage 5 (or stage 6) the computation of the closure set is resumed  whenever $V$ introduces attributes that are new to $X_i$ ($Y_i$). Since the number of the closures considered is $2 |\Sig{IA}|$, we obtain a quadratic time bound for the computation of $Z,X_i,Y_i$.

Let us then present the chase construction for FD+IA.
  The idea is to chase via graphs constructed from vertices and undirected edges,  labeled by sets of attributes (see also \cite{HannulaK14,HannulaK16}). Compared to the traditional tableau chase vertices now represent tuples and labeled edges represent equalities between tuple values. For unrestricted implication the obtained graphical chase procedure is sound and complete; for finite implication it is sound but incomplete.

 \begin{definition}\label{definition:graph}
Let $\Sig{FD}\cup\Sig{IA}\cup\{\si\}$ be a set of FDs and IAs over a relation schema $R$, and let $\Si=\Sig{FD}\cup\Sig{IA}$. Let $G$ be a graph that consists of two vertices $v_0,v_1$ and a single edge $(v_0,v_1)$ that is labeled by $\{A\in R: \Sig{FD}\cup\{X\boto X:XY\boto XZ\in \Sig{IA}\}\models Y \to A\}$ where $Y$ is $U$ if $\si$ is $U\to V$ and otherwise $Y$ is  $\emptyset$. For an attribute set $X$, vertices $v$ and $v'$ are $X$-connected 
 if $v=v'$ or for all $A\in X$ there is a sequence of edges $(v,v_0), \ldots ,(v_n,v')$, each labeled by a set including $A$.  We then denote by $G_{\Si,\si}$ any (possibly infinite) undirected edge-labeled graph that is obtained by chasing $G$ with the following rule as long as possible.
\begin{itemize}
\item Assume that $X\boto Y\in \Si$ and $v,v'$  are two vertices such that no vertex $v''$ is $X$-connected to $v$ and $Y$-connected to $v'$. Then apply (1) once, and thereafter apply (2) as long as possible.
\begin{enumerate}[label=\emph{(\arabic*)}]
\item Add fresh vertex $\vn$ and fresh edges $(v,\vn)$ and $(v',\vn)$ labeled respectively by $X$ and $Y$.
\item For all $X'\to Y'\in \Si$ and $v,v'\in G$ that are $X'$-connected but not $Y'$-connected, add a fresh edge $(v,v')$  labeled by $Y'$.
\end{enumerate}
\end{itemize}
\end{definition}


It is easy to describe non-trivial interaction between FDs and IAs using the above graph construction. Let us illustrate this with an example. Note that in the following example both FDs split IAs from the same assumption set. As a result of this, the interaction between FDs and IAs is already so intricate that all consequences are not derivable using the rules in Table \ref{tab-rules}.

\begin{example}\label{example}
Let $\Si$ be the set
$$\{B\boto CD,D\boto AE,BC\boto ADE,AB\to X,CDE\to X\}$$ and define $\si$ as $A\to X$. Then the basis for $G_{\Si,\si}$ consists of two vertices $v_0,v_1$ that are connected by an $A$-labeled edge. The graph depicted in Fig. \ref{example} is obtained by applying once each IA in $\Si$. Notice that $v_4$ is $AB$-connected to $v_0$ and $CDE$-connected to $v_1$. By the FDs in $\Si$, the next step would be to add $X$-labeled edges $(v_0,v_4)$ and $(v_1,v_4)$.

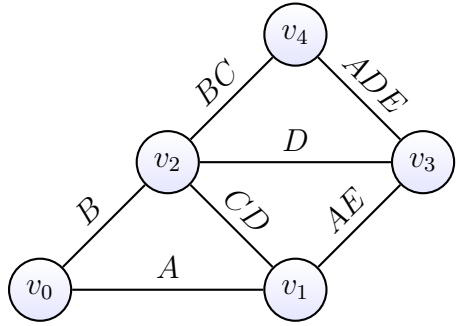
\begin{figure}
\begin{center}
\begin{tikzpicture}[-,>=stealth',auto,node distance=2.4cm,
  thick,main node/.style={circle,top color=white, bottom color=blue!10,draw,font=\sffamily\normalsize\bfseries}]

   \node[main node] (2) {$v_2$};
  \node[main node] (0) [below left of=2] {$v_0$};
  \node[main node] (1) [below right of=2] {$v_1$};
  \node[main node] (3) [above right  of=1] {$v_3$};
  \node[main node] (4) [above right  of=2] {$v_4$};


\path[every node/.style={sloped,anchor=south,auto=false}]

    (0) edge node  {$A$} (1)
    (1) edge node {$CD$}  (2)
    (0) edge node {$B  $}  (2)
    (1) edge node {$AE$}  (3)
    (2) edge node {$D$}  (3)
    (2) edge node {$BC$}  (4)
    (3) edge node {$ADE$}  (4)
;

\end{tikzpicture}
\caption{Start for $G_{\Si,\si}$-construction}
\end{center}
\end{figure}
\end{example}
Since $v_0$ and $v_1$ are $X$-connected in our example $G_{\Si,\si}$,  by the following lemma we obtain that $\Si\models A\to X$. Interestingly, $A\to X$ is not derivable from $\Si$ by the rules depicted in Table \ref{tab-rules}.
\begin{lemma}\label{lemma:help2}

Let $\Si\cup\{\si\}$ be a set of FDs and IAs. Then the following holds:
\begin{enumerate}[label=\emph{(\arabic*)}]
\item if $\si $ is  $U\to V$, then $\Si\models \si$ iff $v_0$ is $V$-connected to $v_1$ in $G_{\Si,\si}$;
\item if $\si $ is $U\boto V$, then $\Si\models \si$ iff there exists $v_2$ that is $U$-connected to $v_0$ and $V$-connected to $v_1$ in $G_{\Si,\si}$.
\end{enumerate}

\end{lemma}
\begin{proof}

Assume first that the right-hand side condition holds; we show how to prove $\Si \models \si$. Let $r$ be a relation satisfying $\Si$, and $t,t'\in r$ be two tuples that agree on $U$ in case (1) and are arbitrary in case (2). It is then an easy induction to show that $\{(v_0,t),(v_1,t')\}$  can be extended to a mapping $f$ from vertices to attributes such that $f(v)(X)=f(v')(X)$ whenever $v$ and $v'$ are $X$-connected in $G_{\Si,\si}$. It follows that $r\models \si$ which shows that $\Si\models \si$.

Assume then that the righ-hand side condition fails.  Then for each vertex $v$ in $G_{\Si,\si}$, construct a tuple $t$ that maps each attribute $A$ to the set that consists of all vertices $v'$ in $G_{\Si,\si}$ that are $A$-connected to $v$. It is  straightforward  to show how this construction gives rise to a (possibly infinite) counter-example for $\Si\models \si$. This completes the proof of the lemma.
\end{proof}

We are now ready to state the non-interaction theorem for FD+IA. 

\begin{theorem}\label{thm:noninter}
Let $\Sig{FD}$ and $\Sig{IA}$ be respectively sets of FDs and IAs, and let $\Sig{FD}^*$ and $\Sig{IA}^*$ be obtained from $\Sig{FD}$ and $\Sig{IA}$ by Algorithm \ref{algZ}.
Then the following holds:
\begin{itemize}
\item if no IA in $\Sig{IA}^*$ splits any  FD in $\Sig{FD}^*$,
then $\Sig{FD}^*$ and  $\Sig{IA}^*$ have no interaction with respect to unrestricted implication;
\item if no IA in $\Sig{IA}$ intersects any  FD in $\Sig{FD}$, then $\Sig{FD}$ and  $\Sig{IA}$ have no interaction with respect to finite implication.
\end{itemize}
\end{theorem}
\begin{proof}
We consider first unrestricted implication and then finite implication.\\
\textbf{Unrestricted case.}  Assume that $\Si\models \si$ where $\Si =\Sig{FD}^*\cup \Sig{IA}^*$, and assume that no IA in $\Sig{IA}^*$ splits any FD in $\Sig{FD}^*$; we show that $\Sig{IA}^*\models \si$ or $\Sig{FD}^*\models \si$ holds. Let $G_{\Si,\si}$ be as in Definition \ref{definition:graph}; we show that the construction of $G_{\Si,\si}$ does not use any application of (2). 
Assume to the contrary that some introduction of a new vertex $\vn$ and new edges $(v,\vn),(v',\vn)$ by (1) renders for the first time some vertices $w$ and $w'$  $X'$-connected but not $Y'$-connected, for some  $X'\to Y'\in \Si$. Assume that the new edges $(v,\vn)$ and $(v',\vn)$ are respectively labeled by $X$ and $Y$ in which case we have $X\boto Y\in \Si$, and assume first that $w'$ is $\vn$. Since $X\boto Y$ does not split $X'\to Y'$, it must be the case that $X'$ is included in either $X$ or $Y$. We may assume by symmetry that this holds for $X$ in which case $X'Y'\sub X$ by (i). Then it follows that $w$ and $v$ are $X'$-connected but not $Y'$-connected since the same holds for $w$ and $\vn$ by the assumptions. Now $w\neq \vn$, and hence this contradicts Definition \ref{definition:graph} if $\vn$ is the first added vertex, and otherwise the assumption that (2) has not been applied previously. Hence, $w'$ and $\vn$ are different vertices, and so by symmetry are $w$ and $\vn$. However, since $X$ and $Y$ share only constant attributes, $w$ and $w'$ must have been $X'$-connected already before the introduction of $\vn$. Hence by Definition \ref{definition:graph} they have also been $Y'$-connected  which contradicts the assumption. This completes the proof that no  application of (2) occurs in the construction of $G_{\Si,\si}$.

Assume  that $\si$ is an IA. By the construction of $\Sig{IA}^*$,  $(v_0,v_1)$ has the same label in the initial graph $G$ of $G_{\Si,\si}$ and $G_{\Sig{IA}^*,\si}$. Since  the construction of $G_{\Si,\si}$  contains no application of (2),  $G_{\Si,\si}$ equals $G_{\Sig{IA}^*,\si}$. Hence by Lemma \ref{lemma:help2},  $\Sig{IA}^*\models \si$.
Assume then that $\si$ is an FD. Since the graph construction of $G_{\Si,\si}$ does not introduce any new labels for $(v_0,v_1)$, we obtain that $\Sig{FD}^*\cup\{X\boto X:XY\boto XZ\in \Sig{IA}^*\}\models \si$ by Definition \ref{definition:graph} and Lemma \ref{lemma:help2}. Therefore, and since $Z$ satisfies (ii) and $\emptyset \to Z \in \Sig{FD}^*$, it follows  that $\Sig{FD}^*\models \si$. \\
\textbf{Finite case.}  Assume  that no IA in $\Sig{IA}$ intersects any FD in $\Sig{FD}$, and assume first that $\Sig{FD}\not\fmodels U\to V$; we show that $\Sig{FD}\cup\Sig{IA}\not\fmodels U\to V$.  Let $U^+$ be the set of  attributes $A$ such that $\Sig{FD}\fmodels U\to A$, and let $I$ be the set of  attributes  that appear in $\Sig{IA}$. Then we let $r$ be a relation where $r(U^+)$ is ${}^{U^+}\{0\}$, $r(I\setminus U^+)$ is ${}^{I\setminus U^+}\{0,1\}$, and $r(A)$ takes values $1, \ldots ,2^{|U^+\setminus I|}$ for $A\in R\setminus (U^+\cup I)$, where  $R$ is the underlying relation schema. Clearly $r$ satisfies $\Sig{IA}$ and violates $U\to V$. If $X\to Y\in \Sig{FD}$, then by the assumption   $X\cap I=\emptyset $, and hence it follows by the construction that $r$ satisfies $X\to Y$.

Assume then that $\Sig{IA}\not\fmodels U\boto V$; we show that $\Sig{FD}\cup\Sig{IA}\not\fmodels U\boto V$. Let $I$ be as in the previous paragraph. If $UV\sub I$, let $r$ be a finite relation over $I$ satisfying $\Sig{IA}$ and violating $ U\boto V$; otherwise define $r$ as ${}^{I}\{0,1\}$. Then let $r'$ be the extension of $r$ where, for each $A\in R$, $r'(A)$ takes values $1, \ldots ,2^{|r|}$. Given the non-interaction assumption, we notice that $r'$ is a witness of $\Sig{FD}\cup\Sig{IA}\not\fmodels U\boto V$.  This completes the proof.
\end{proof}
Note that the finite relation constructions in the proof are possible already from the assumptions that $\Sig{FD}\not\models U\to V$ or $\Sig{IA}\not\models U\boto V$; for the latter  recall that finite and unrestricted implication coincide for IAs. Hence, the proof entails  that the finite and unrestricted implication problems coincide for FD+IA provided that the non-intersection assumption holds.

To illustrate the necessity for a stronger condition in the finite case, recall from Section \ref{sect:FD+IA} that $AB \to CD$ is finitely implied by $\{A\boto B, C\boto D, BC\to AD, AD\to BC\}$, and notice  that $AB \to CD$ is not finitely implied by  $\{ BC\to AD, AD\to BC\}$. However, Algorithm \ref{algZ} does not produce any fresh  assumptions, and neither $A\boto B$ nor $C\boto D$ splits any FD assumption. Therefore, lackness of splits is not sufficient for non-interaction in the finite case.

\section{Complexity Results}\label{sect:comp}

In this section we examine the computational complexity of the various implication problems we have studied. We first show that both implication problems for UFD+UIND+IA can be solved in low-degree polynomial time, even though the problems differ from one another. Then we focus on the class IND+IA for which the two implication problems coincide. Recall that the combination of FDs and INDs is undecidable with regards to their implication problem. However, the same cannot be true for the combination of INDs and IAs, as already witnessed by our finite axiomatization. In this section we proceed even further by showing that adding IAs to the class of INDs involves no trade-off in terms of losing desirable computational properties. Indeed, Theorem  \ref{comp:IND+IA} shows that, alike for INDs \cite{CasanovaFP84journal}, the  implication problem for IND+IA is $\PSPACE$-complete. We start by analyzing the complexity of UFD+UIND+IA implication.

\begin{theorem}
\label{thm:compunary}
Let $\Sig{UFD},\Sig{UIND},\Sig{IA}$ be respectively sets of UFDs, UINDs, and IAs over a relation schema $R$. The unrestricted and finite implication problems for $\si$ by $\Sig{UFD}\cup\Sig{UIND}\cup\Sig{IA}$ can be decided in time:
\begin{itemize}
\item $O(|\Sig{IA}|\cdot  |\Sig{UFD}|+|\Sig{UIND}|)$ if $\si$ is an UFD or UIND;
\item $O(|\Sig{IA}|\cdot ( |\Sig{UFD}|+|R|^2)+ |\Sig{UIND}|)$ if $\si$ is a IA.
\end{itemize}
\end{theorem}

\begin{proof}
 Algorithm~\ref{algZ2} extends an algorithm for UFD+UIND-implication in \cite{CasanovaFP84journal}. It generates a graph $G$ and sets $Z,X_i,Y_i$, for $i=1, \ldots ,|\Sig{IA}|$, and takes $\Sig{UFD},\Sig{UIND}$ and $\Sig{IA}=\{U_1\boto V_1, \ldots ,U_n\boto V_n\}$ as an input. Note that steps 2-4 are to be omitted in the unrestricted case.

\begin{algorithm}[t]
\caption{Algorithm for computing $Z,X_i,Y_i$}
\label{algZ2}

\begin{algorithmic}[1]
\Require $\Si_{\rm UFD}$, $\Si_{\rm UIND}$, and $\Si_{\rm IA}=\{U_i\boto V_i\mid i=1, \ldots ,n\}$
 \Ensure A digraph $G$ and sets $Z$ and $X_i, Y_i$ for  $i=1, \ldots ,n$
 \State \textbf{Initialize:}   $Z,X_i,Y_i$ are empty and $G$ is a digraph that consists of vertices $R$, red edges $(A,B)$ for $A\to B\in \Sig{UFD}$, and black edges $(A,B)$ for $B\sub A\in \Sig{UIND}$
			\State  \textbf{compute} all strongly connected components of $G$ (\emph{only in the finite case})
			\For{each red (black) edge $(A,B)$ with $A,B$ in the same component (\emph{only in the finite case})} 		 \State \textbf{add} red (black) edge $(B,A)$ to $G$
			\EndFor
			\For{ $\emptyset \to A \in \Sig{FD}$}
				\State \textbf{add} $A$ to $Z$
			\EndFor
			\For{$n=1, \ldots ,n$}
				\State \textbf{add} $A$ to $X_i$ ($Y_i$) if there is a red path from $X_i$ ($Y_i$) to $A$
				\State \textbf{add} each attribute in $X_i\cap Y_i$ to $Z$
			
			\EndFor
			\State \textbf{add} to $Z$ all attributes that are reachable from $Z$ in $G$, ignoring edge colors
			\For{each red (black) edge $(A,B)$ with $A,B$ in $Z$}
				\State \textbf{add} red (black) edge $(B,A)$ to $G$
			\EndFor
\end{algorithmic}
\end{algorithm}

Let $\Si:=\Sig{UFD}\cup\Sig{UIND}\cup\Sig{IA}$. Let $\Si^*$ be the set of dependencies that consists of all trivial UFDs, UINDs, and IAs over $R$, and:
\begin{enumerate}[label=\emph{(\roman*)}]
\item $A\to B$ iff  $B\in Z$ or $A$ is connected to $B$ by a red path;
\item $A\sub B$ iff there is a black path from $B$ to $A$;
\item $X\boto Y$ iff $\Sig{IA}^*\vdash_{\mathfrak{I}} X\boto Y$;
\end{enumerate}
where $\Sig{IA}^*:=\{X_i\boto Y_i\mid i=1, \ldots ,n\}\cup\{Z\boto Z\}$.

Let us first consider the case for finite implication. In what follows, we show that $\Si\fmodels \si$ iff $\si\in \Si^*$. By Theorem \ref{thm:finaxioms}, $\Si\fmodels\si$ iff  $\si$ can be deduced from $\Si$ by rules $\mathfrak{A}^*\cup \{\iax{1},\iax{2}\}\cup\{\cax{k}: k\in \N\}$. Therefore, the claim follows if $\Si^*$ is the deductive closure of $\Si$ under these rules. It is straightforward to check that each dependency in $\Si^*$ can be deduced by the rules. Note that item 3 of the graph construction can be simulated with the cycle rules and transitivity rules for FDs and INDs.
Next we show that $\Si^*$ is deductively closed. The only non-trivial cases are the cycle rules and $\fiax{1},\fiax{2}$. For $\fiax{2}$, assume that $X\boto YA,A\to B\in \Si^*$; we show that $X\boto YAB\in \Si^*$. The cases where $A$ or $B$ is empty are trivial. Assume that both are single attributes.  If $B\in Z$, then the claim follows by the definition of $\Sig{IA}^*$. Otherwise, $B\not\in Z$ and there exists a red path from $A$ to $B$. 
 Since $Z$ is closed under red arrows, this path stays outside $Z$ and has  hence existed already at step $7$. Therefore by the construction $B$ is in $X_i$ ($Y_i$)  whenever $A$ is. It is now an easy induction on the length of a deduction to show that whenever $\Sig{IA}^*\vdash_{\mathfrak{I}} V_1\boto V_2$, then $\Sig{IA}^*\vdash_{\mathfrak{I}} V'_1\boto V'_2$ where $V'_i=V_iB$ if $A\in V_i$, and otherwise $V'_i=V_i$. From this it follows that $\Sig{IA}\vdash_{\mathfrak{I}} X\boto YAB$, and therefore $ X\boto YAB\in \Si^*$. For $\fiax{1}$ and the cycle rules, the reasoning is analogous. This concludes the proof of the claim for finite implication.

Let us then turn to unrestricted implication. It suffices to show that $\Si\models \si$ iff $\si\in \Si^*$, where  $\Si^*$ is now defined over graph $G$ obtained from steps 1,5-12 of the algorithm. Proving this is analogous to the finite case (with the exception that rules $\uax{3},\uax{4}$ are to be considered instead of the cycle rules) and hence omitted here. 

We now analyze the time complexity. Steps $1,2,3,10,11$ each take time $O(|\Sig{UFD}|+|\Sig{UIND}|)$, step $5$ takes $O(|\Sig{UFD}|)$, and step $7$ takes $O(|\Sig{IA}|\cdot  |\Sig{UFD}|)$. Since, reachibility can be tested in linear time, we obtain the time bound for a UFD or a UIND $\si$. Assume that $\si$ is an IA. It easy to see that $\si$ is (finitely) implied by $\Sig{IA}^*$ iff $\si\upharpoonright R'$ is (finitely) implied by $\Sig{IA}^*\upharpoonright R'$, where $R':=R\setminus Z$. On the other hand, by Theorem 2 in \cite{KontinenLV13} the right-hand side implication problem for disjoint independence atoms can be decided in time $O(|\Sig{IA}^*\upharpoonright R'|\cdot |R'|^2)$. Since, $\Sig{IA}^*\upharpoonright R'$ is of the size of $\Sig{IA}$, the time bound for an IA $\si$ follows.
\end{proof}

It follows that all unary INDs and constancy atoms implied by a set of INDs and IAs can be recognized in linear time.
\begin{theorem}\label{thm:uind+ca+linear}
The unrestricted and finite implication problems for the class CA+UIND by IND+IA is linear-time decidable.
\end{theorem}
\begin{proof}
Let $\Si$ be a set of INDs and IAs, and let $\sigma$ be a UIND and $\tau$ a CA. By Theorem \ref{theorem:uind+uca}, $\Si\models \rho \Leftrightarrow \Sigma_{\rm UIND} \cup\Sigma_{\rm CA}\models \rho$, for $\rho\in\{\sigma,\tau\}$.  Let $Y$ be the singlevalued span of $\Sigma_{\rm UIND} \cup\Sigma_{\rm CA}$ described in Definition \ref{span}, and let  $\Sigma_0:= \Sigma_{\rm UIND}\cup \{A\sub B: B\sub A \in \Sigma_{\rm UIND}, A\in Y\}$ and $\Sigma_1:=\Sigma_{\rm CA}\cup \{\emptyset \to A:A\in Y\}$.
By Theorem \ref{thm:cosma},
\begin{itemize}
\item $\Sigma_{\rm UIND} \cup\Sigma_{\rm CA}\models \sigma \Leftrightarrow \Sigma_0 \models \sigma$,
\item $\Sigma_{\rm UIND} \cup\Sigma_{\rm CA}\models \tau \Leftrightarrow \Sigma_1 \models \tau $.
\end{itemize}
 The singlevalued span $Y$ and the deductive closure of $\Si_0$ can be computed in linear time by reducing to graph reachability. For the latter, note that  $\iax{1},\iax{2}$ form a complete axiomatization  for UINDs. Moreover, $\Si_1\models  \emptyset \to A$ iff $A\in Y$.  We conclude that in both cases implication can be tested in linear time.
\end{proof}

Next we turn to the class IND+IA and use graphs again to characterize the associated implication problem. Recall by Corollary \ref{cor:coincide} that  the unrestricted and finite implication problems  coincide for IND+IA.
\begin{definition}\label{def:graph}
Let $\Sigma\cup\{\sigma\}$ be a set of INDs and IAs, and assume that $\sigma$ is of the form $R[ X]\sub S[ Y]$ (or $R[ X_1\boto   X_2]$ where $ X= X_1 X_2$). Then we let  $H_{\Si,\si}$  be a graph that has nodes $\{\tau_1, \ldots , \tau_{k}\}$, for $k\leq | X|$, where $\tau_i$ is an IND of the form $R[ A_i]\sub R'[ B_i]$ and the concatenation $ A_1 \ldots  A_n$ is a permutation (without repetition) of $ X$. Two nodes $v,v'$ are connected by a directed edge $v\to v'$ if one of the following holds:
\begin{itemize}
\item There exists $\tau_0,\tau_1,\tau$ such that $v\setminus\{\tau\}=v'\setminus\{ \tau_0,\tau_1\}$ where $\tau$ is a permutation of $R[U_0U_1]\sub R'[V_0V_1]$, $\tau_0=R[U_0]\sub R'[V_0]$, $\tau_1=R[U_1]\sub R'[V_1]$.
\item There exists $\tau_0,\tau_1,\tau$ such that $v\setminus\{\tau_0,\tau_1\}=v'\setminus\{\tau\}$ where $\tau_0=R[U_0]\sub R'[V_0]$, $\tau_1=R[U_1]\sub R'[V_1]$, $\tau=R[U_0U_1]\sub R'[V_0 V_1]$, and 
 for some $ W_0\supseteq V_0$ and
$ W_1\supseteq V_1 $, $R'[W_0\boto W_1]\in \Sigma$.
\item There exists $\tau,\tau'$ such that $v\setminus\{\tau\}=v'\setminus\{ \tau'\}$ where $\tau =R[U]\sub R'[V],\tau'=R[U]\sub R''[W]$,   and 
$R'[V]\sub R''[W]$ is a projection and permutation of some IND in $\Sigma$.
\end{itemize}
If $\sigma$ is  $R[ X_1\boto  X_2]$, then we define $v_{\rm start}:=\{R[ X_1]\sub R[ X_1],R[ X_2]\sub R[ X_2]\}$ and $v_{\rm end}:=\{R[ X]\sub R[ X]\}$. If  $\sigma$ is  $R[ X]\sub S[ Y]$, then $v_{\rm start}:=\{R[ X]\sub R[ X]\}$ and $v_{\rm end}:=\{R[ X]\sub S[ Y]\}$.
\end{definition}
\begin{lemma}
Let $\Sigma\cup\{\sigma\}$  be a set of INDs and DIAs. Then $\Sigma\models \sigma$ iff $H_{\Si,\si}$ contains a directed path from $v_{\rm start}$ to $v_{\rm end}$.
\end{lemma}
\begin{proof}
Assuming $\Sigma\models \sigma$, the required path is found by backtracking a succesful chase of $d$ by $\Sigma$ where the chase rules and $d$ are defined as in cases 2) and 3) in the proof of Theorem \ref{theorem:ind+ia}.

For the other direction, let $d$ be a database satisfying $\Sigma$. Assume first that $\sigma$ is an IA of the form $R[ X_1\boto  X_2]$, and let $t,t'\in r[R]$. Now $ X_1$ and $ X_2$ are disjoint, so we can define  a mapping $t_0$ that agrees with $t$ on $ X_1$ and with $t'$ on $ X_2$. It  suffices to show, given a directed path in $H_{\Si,\si}$ from $v_{\rm start}$ to $\{\tau_1, \ldots , \tau_{k}\}$, that for each $\tau_i$ of the form $R[ U_i]\sub R'[V_i]$ there is $t_i\in r'[R']$ such that $t_0(U_i)=t_i(V_i)$. Since this is a straightforward induction, we  leave the proof to the reader. The case where $\sigma$ is an IND is analogous.
\end{proof}

$\PSPACE$-completeness of the IND+IA-implication is now showed by reducing to graph reachability in $H_{\Si,\si}$.
\begin{theorem}\label{comp:IND+IA}
The unrestricted (finite) implication problem for IND+IA is complete for $\PSPACE$.
\end{theorem}
\begin{proof}
The lower bound follows by the fact that the implication problem for INDs alone is $\PSPACE$-complete  \cite{CasanovaFP84journal}. For the upper bound, let  $\Si\cup \sigma$ be a set of INDs and IAs where $\sigma$ is an IA (or an IND). Construct first $\Sigma_0\cup\{\sigma_0\}$ ($\Sigma_1$) as described in Lemma \ref{help}. By Theorem \ref{thm:uind+ca+linear} this can be done in polynomial time. Then non-deterministically check whether there is a directed path in $H_{\Sigma_0, \sigma_0}$ ($H_{\Sigma_1, \sigma}$) from $v_{\rm start}$ to $v_{\rm end}$. Since this requires only polynomial amount of space, we conclude by Savitch' theorem that the implication problem is in $\PSPACE$.
\end{proof}

Note that there are only polynomially many nodes in  $H_{\Si,\si}$, given that $\sigma$ is of fixed  arity. Hence, by Lemma \ref{help} and Theorem \ref{thm:uind+ca+linear} we obtain the following corollary.
\begin{corollary}
The unrestricted (finite) implication problem for $\sigma$ by $\Sigma$, where $\Sigma\cup \sigma$ is a set of INDs and IAs, is fixed-parameter tractable in the arity of $\sigma$.
\end{corollary}
Actually, the choice of the parameter in the above corollary is not optimal. By Theorem \ref{thm:uind+ca+linear}, tractability is preserved if only the number of non-constant attributes  in an IA $\sigma$ is fixed. Moreover, assume that $\sigma$ is an IND of the form $A_1\ldots A_{h+k} \sub B_1\ldots B_{h+k}$ where $B_i$ is constant for $i\leq h$. Then the implication problem for $\sigma$ is fixed-parameter tractable in $k$ since by $\uax{1}$ it suffices to test whether $A_{h+1}\ldots A_{h+k}\sub B_{h+1}\ldots B_{h+k}$ and $A_i\sub B_i$, for $i\leq h$, are all implied.

\section{Conclusion and Outlook}\label{sec:conclusion}

In view of the infeasibility of EMVDs and of FDs and INDs combined, the class of FDs, MVDs and unary INDs is important as it is low-degree polynomial time decidable in the finite and unrestricted cases. As independence atoms form an important tractable embedded subclass of EMVDs, we have delineated axiomatisability and tractability frontiers for subclasses of FDs, INDs, and IAs. The most interesting class is that of IAs, unary FDs and unary INDs, for which finite and unrestricted implication differ but each is axiomatisable and decidable in low-degree polynomial time. The subclass is robust with this properties as unary functional and binary inclusion dependencies are undecidable in both the finite and unrestricted case, and binary functional dependencies with unary independence atoms are not finitely axiomatisable in the finite. The results form the basis for new applications of these data dependencies in many data processing tasks.

Even though research in this space has been rich and deep, there are many problems that warrant future research. Theoretically, the decidability remains open for both independence atoms and functional dependencies as well as unary independence atoms and functional dependencies, both in the finite and unrestricted case. 
This line of research should also be investigated in the probabilistic setting of conditional independencies, fundamental to multivariate statistics and machine learning.
Practically, implementations and experimental evaluations of the algorithms can complement the findings in the research. It would be interesting to investigate applications. For example, in consistent query answering the aim is to return all those answers to a query that are present in all repairs \cite{DBLP:series/synthesis/2011Bertossi,DBLP:conf/pods/KoutrisW15}. Here, a repair is a database obtained by applying a minimal set of operations that resolve all violations of the given database with respect to the given constraints). Assuming that the constraint is the IA $\sigma_{11}=\textsc{Heart}:\;p\_name\boto t\_id$, operations are tuple insertions (appropriate for tuple-generating dependencies), and the query is
\begin{center}
\begin{tabular}{l}
SELECT p\_name, t\_id FROM \textsc{Heart}\,;
\end{tabular}
\end{center}
a rewriting of this query that would return the consistent query answers over the given database is
\begin{center}
\begin{tabular}{l}
SELECT H.p\_name, H'.t\_id \\
FROM \textsc{Heart} H, \textsc{Heart} H'\,;
\end{tabular}
\end{center}
It would be interesting to include IAs when investigating typical problems of consistent query answering \cite{DBLP:series/synthesis/2011Bertossi}. In database design \cite{DBLP:journals/is/GottlobPW10,koehler:2016} IAs are useful for finding lossless decompositions of a database schema, which is exemplified in \cite{DBLP:journals/tods/Delobel78}. Considering the infeasibility of the implication problem for EMVDs, and the limited knowledge on the interaction of IA+FD, it would be interesting to investigate the possibilities for automating schema design using subclasses of IA+FD+IND. Query folding has greatly benefited from considering FD+IND~\cite{DBLP:journals/is/Gryz99}. As the implication problems for FD+IND are undecidable, there is no algorithm that can produce a complete list of query rewritings for this class. However, our results are a starting point to develop (complete) query folding algorithms for subclasses of IA+FD+IND. Similarly, well-known methods for deciding query containment in the presence of FD+IND~\cite{DBLP:journals/jcss/JohnsonK84} could be extended to subclasses of IA+FD+IND. Query answering has also been shown to be effective on data integration systems in which global schemata are expressed by key and foreign key constraints \cite{DBLP:books/daglib/p/CaliC13}. It would be interesting to see how these techniques can be extended to other subclasses of IA+FD+IND. Another area of impact for our results is ontology-based data access, in which keys and inclusion dependencies play an important role. It would be interesting to study to which degree independence atoms can be added without increasing too much the complexity of associated decision problems \cite{DBLP:conf/rweb/GottlobMP15}. It is interesting to develop algorithms that discover all those IAs that hold on a given database \cite{DBLP:journals/vldb/AbedjanGN15}. For subclasses of IA+FD+IND, our algorithms can remove data dependencies implied by others.

\bibliographystyle{plain}
\bibliography{biblio2}
\end{document}